\documentclass{article}
\usepackage{graphicx} % Required for inserting images
\usepackage{amsmath}
\usepackage{amssymb}

\usepackage[letterpaper,top=2cm,bottom=2cm,left=3cm,right=3cm,marginparwidth=1.75cm]{geometry}

\usepackage{adjustbox}
\usepackage[english]{babel}
\usepackage[T1]{fontenc}
\usepackage{url}
\usepackage{xcolor}
\usepackage{tikz}
\usetikzlibrary{positioning, fit}

\newcommand{\bv}{\mathrm{BV}}

\newcommand{\f}{{\ensuremath{\varphi}}}

\newcommand{\h}{{\ensuremath{\eta}}}

\renewcommand{\leq}{\leqslant}
\renewcommand{\geq}{\geqslant}
\renewcommand{\le}{\leqslant}
\renewcommand{\ge}{\geqslant}
\renewcommand{\to}{\rightarrow}

\newcommand{\fer}[1]{(\ref{#1})}
\def\be#1\ee{\begin{equation}#1\end{equation}}
\newenvironment{equations}{\equation\aligned}{\endaligned\endequation}
\newcommand{\R}{\mathbb R}
\newcommand{\N}{\mathbb N}

\newcommand{\floor}[1]{\left\lfloor #1 \right\rfloor}

\def\f{\hat f}

\def\DDSS{\mathcal{D}_\mathcal{S}}

\def \Id{\mathbf{I}}
\def\be#1\ee{\begin{equation}#1\end{equation}}

\def\bqa{\begin{eqnarray}}
\def\eqa{\end{eqnarray}}

\def\bZ{{\bf Z}}

\newcommand{\bd}{\begin{displaymath}}
\newcommand{\ed}{\end{displaymath}}
\newcommand{\ba}{\begin{eqnarray}}
\newcommand{\ea}{\end{eqnarray}}

\def\N{\mathbb{N}}
\def\R{\mathbb{R}}

\def\bx{{\bf x}}
\def\bX{{\bf X}}
\def\bY{{\bf Y}}
\def\bV{{\bf V}}
\def\bW{{\bf W}}
\def\by{{\bf y}}
\def\bm{{\bf m}}
\def\bu{{\bf u}}
\def\bv{{\bf v}}
\def\bw{{\bf w}}
\def\bxi{{\pmb\xi}}
\def\bbeta{{\pmb\beta}}

\DeclareMathOperator{\diag}{diag}
\DeclareMathOperator{\GMD}{GMD}

\def \PP{\mathcal{P}}

\def \simS{\sim_{\mathcal{S}}}

\def \erre{\mathbb{R}}

\def \enne{\mathbb{N}}

\usepackage{hyperref}
\usepackage{booktabs}
\usepackage[inline]{enumitem}
\usepackage{verbatim}

\title{FROM KINETIC THEORY TO AI: A REDISCOVERY OF HIGH-DIMENSIONAL DIVERGENCES AND THEIR PROPERTIES
}
\author{
  Gennaro Auricchio\thanks{Department of Mathematics, University of Padova. \texttt{gennaro.auricchio@unipd.it}} \and
  Giovanni Brigati\thanks{Institute of Science and Technology Austria. \texttt{giovanni.brigati@ist.ac.at}} \and
  Paolo Giudici\thanks{Department of Economics and Management, University of Pavia (Italy). \texttt{paolo.giudici@unipv.it}} \and
  Giuseppe Toscani\thanks{Department of Mathematics, University of Pavia. \texttt{giuseppe.toscani@unipv.it}}
}

\newtheorem{theorem}{Theorem}
\newtheorem{remark}{Remark}
\newtheorem{definition}{Definition}
\newtheorem{proof}{Proof}

\begin{document}

% \markboth{G. Auricchio, 
% G. Brigati, 
% P. Giudici, G. Toscani}{High--dimensional divergences and their properties}

%%%%%%%%%%%%%%%%%%% Publisher's Area please ignore %%%%%%%%%%%%%%%%%%%%%%%
%
% \catchline{}{}{}{}{}
%
%%%%%%%%%%%%%%%%%%%%%%%%%%%%%%%%%%%%%%%%%%%%%%%%%%%%%%%%%%%%%%%%%%%%%%%%%%

\maketitle

% \begin{center}
% \begin{minipage}{0.4\textwidth}
% \centering
% \textbf{Gennaro Auricchio} \\
% Department of Mathematics \\
% University of Padova \\
% \texttt{gennaro.auricchio@unipd.it}
% \end{minipage}
% \hspace{0.1\textwidth}
% \begin{minipage}{0.4\textwidth}
% \centering
% \textbf{Giovanni Brigati} \\
% Institute of Science and Technology Austria \\
% \texttt{giovanni.brigati@ist.ac.at}
% \end{minipage}

% \vspace{1em} % space between rows

% \begin{minipage}{0.4\textwidth}
% \centering
% \textbf{Paolo Giudici} \\
% Department of Economics and Management \\
% University of Pavia (Italy) \\
% \texttt{paolo.giudici@unipv.it}
% \end{minipage}
% \hspace{0.1\textwidth}
% \begin{minipage}{0.4\textwidth}
% \centering
% \textbf{Giuseppe Toscani} \\
% Department of Mathematics \\
% University of Pavia \\
% \texttt{giuseppe.toscani@unipv.it}
% \end{minipage}
% \end{center}

\begin{abstract}
Selecting an appropriate divergence measure is a critical aspect of machine learning, as it directly impacts model performance. Among the most widely used, we find the Kullback–Leibler (KL) divergence, originally introduced in kinetic theory as a measure of relative entropy between probability distributions. Just as in machine learning, the ability to quantify the proximity of probability distributions plays a central role in kinetic theory. In this paper, we present a comparative review of divergence measures rooted in kinetic theory, highlighting their theoretical foundations and exploring their potential applications in machine learning and artificial intelligence.

\end{abstract}

\textbf{keywords.}{ Kinetic Theory, Artificial Intelligence, Machine Learning, Divergence measures, Multivariate distributions, Scale Invariance, Wasserstein Distance, Gini index. }

\vskip.2cm


%\tableofcontents

%\vskip.5cm
%\textbf{Remarks addressed all authors are in bold. Those addressed specifically to Andrea and Giuseppe are in blue. Those addressed specifically to Giulia and Lorenzo are in red. Remarks by Nino are in green (English). These remarks are simply suggestions and are NOT mandatory}
%\vskip.5cm
%\vskip.5cm
%{\color{green}
%\textbf{For the moment each sub-team cares of a specific section-topic. Corrections will be done when the whole paper has been written,}}
%\vskip.5cm

%%%%%%%%%%%%%%%%%%%%%%%%%%%%%%%%%%%%%%%%%%%%%%%%%%%%%%%%
%%%%%%%%%%%%%%%%%%%%%%%%%%%%%%%%%%%%%%%%%%%%%%%%%%%%%%%%
%\textcolor{purple}{[GB:
%\begin{itemize}
 %   \item check (for myself) that the idea of Paolo about measuring inequality w.r.t.~a physical average VS measuring cross-variation is highlighted. 
  %  \item missing reference to the EU quest \url{https://composite-indicators.jrc.ec.europa.eu/multidimensional-inequality}.
   % \item question: does the multi-dimensional Gini write in terms of one-dimensional indicators? \textcolor{blue}{I'd say YES if the measures are independent, NO otherwise, but it holds with an $\leq$}
    %\item question: shall I make the comparison Fourier VS Gini more explicit? On the other hand, as the proof will show, there is not so much to do\dots
    %\item question: I wrote everything for the Mahalanobis case. Shall I say somewhere that many ideas carry over to any whitening process? 
%\end{itemize}]}

\section{Introduction}
In machine learning, a \emph{divergence}\footnote{While divergences have a specific meaning in statistics, we use the term more broadly to denote any function that compares pairs of probability measures. This avoids introducing unnecessary terminology.} function quantifies the discrepancy between a model's predictions and the ground truth. 
It serves as a measure of model performance: small divergence indicates accurate predictions, while large divergence signals greater error.
% 
% In machine learning, \emph{divergence}\footnote{Albeit divergences have a specific meaning in statistics, throughout the paper, we use the term divergence to denote a generic function able to compare any pair of probability measures. We decided to opt for this solution to avoid the introduction of new unnecessary terms.} functions are employed to measure model performance by calculating the deviation of model predictions from the "correct" predictions, the \emph{ground truth}. 
% % 
% In other words, a divergence function tracks the degree of error in the results of a statistical machine learning model, embedded in an artificial intelligence (AI) application. It does so by quantifying the difference between the predictions (the model outputs for the given inputs) and the actual expected value, or ground truth. 
% % 
% If the model's predictions are accurate, the divergence is small.
% % 
% If its predictions are inaccurate, the divergence is large.
% 
In the statistical machine learning literature, there exists a wide variety of divergence functions, each suited to different objectives, data types, and priorities. 
In most cases, however, the divergence is calculated in terms of \emph{relative entropy}, as the entropy is a standard measure of uncertainty within a system. 
Among the many definition of entropy, the one usually adopted is the Kullback-Leibler divergence (KL divergence) introduced by Boltzmann and Gibbs, and widely used in information theory,\cite{KL} which measures the difference between two probability densities $f$ and $g$ through the formula
\[
    H(f|g) = \int_{\erre^n} f(\bx)\log\frac{f(\bx)}{g(\bx)}d\bx.
\]
The fact that the Kullback-Leibler divergence is related to the concept of relative entropy establishes in a natural way a close connection with the kinetic theory of rarefied gases, and, more in general, with statistical physics.
This connection, which allows to make use of various concepts originally developed for physically relevant problems, has been recently remarked in Refs. \cite{BDL,BE} in relationship with the modelling and simulation of the collective dynamics of a system composed by a large number of interacting living entities, \textit{id est} a multi-agent system.

%
% In multi-agent systems, modelling involves studying the behaviour of systems composed of a multitude of interacting agents. 
% % 
% In particular, the physical description of the evolution of how social traits (such as wealth, opinion, knowledge, etc.) are distributed in a multi-agent society received a lot of attention.\cite{YR,CFL,NPT,Sin,PT13,Chak,Chak1,Sen} 
% 

% 
Multi-agent systems are particularly well-suited to study the behaviour of models composed of a multitude of interacting agents. 
In particular, these models have been used to retrieve a physical description of how the distribution of social traits (such as wealth, opinion, knowledge, etc.) evolves.\cite{YR,CFL,NPT,Sin,PT13,Chak,Chak1,Sen} 
The expected behaviour of these physical driven multi-agents system is that, as a result of the action of the known microscopic interactions between agents, a universal equilibrium configuration typical of the phenomenon under study emerges.
One of the main features of such systems is their tendency to converge rapidly toward an equilibrium state when isolated, and, interestingly, the convergence to the equilibrium is often related to a thermodynamic principle.
Indeed, it was to describe the macroscopic behaviour of a large number of frequently interacting molecules, that the entropy functional was originally introduced in the late 19th century.
In fact, according to Gibbs' principle, \emph{the equilibrium distribution is the one that maximizes entropy subject to the constraints imposed by the relevant conservation laws.}
There are several mathematical models describing how a distribution converges (or \textit{relaxes}) to an equilibrium, of whom the spatially homogeneous Boltzmann equation is the most famous, both for its important current applications and for historical reasons.\cite{Bol} 
Indeed, the Boltzmann's $H$-theorem on increasing entropy was the first analytical demonstration of the second principle of thermodynamics.
As a system in equilibrium depends only on a limited number of unknowns (certain macroscopic observables), the identification of equilibrium configurations and the speed of convergence towards them are of paramount importance for applications.
Knowing the convergence rate is the only way to understand the relevant time scale for the equilibration process, which, in turn, is important for modelling and for assessing the feasibility of numerical simulations. 
This problem has a long history, which has led to a number of mathematical results based on suitable divergences that quantify how the distance between the solution and its equilibrium configuration evolves through time.\cite{AA} 
A careful analysis of the research methodologies aimed at measuring how close the solution of the kinetic problem is from its equilibrium distribution leads to the conclusion that their performance can be directly interpreted through divergences, as these items describe the differences between distributions. 
In this paper, we highlight how the study of relaxation to equilibrium in kinetic models naturally leads to the introduction of various divergence measures between probability distributions, illustrating the connection between kinetic theory and the field of artificial intelligence.
Building on this connection, we focus on three physics-inspired families of divergences: the Wasserstein Distances, the Fourier-based Metrics, and the Energy Distances. 
% \textcolor{purple}{[GB: I would write \emph{distances, metrics}.]}
%
Motivated by the question of whether these divergences are valuable tools within the context of artificial intelligence, we investigate their properties and examine their relationships.
Moreover, we show that these divergences can be modified to make them scale invariant, an important property that is necessary to treat multi-dimensional data collected with respect to different units of measure.\cite{ABGT}
To better understand the physical reasons that led to the selection of a small number of divergences to measure how close two distributions are, we start our analysis with a brief survey about the problem of convergence towards equilibrium of a spatially homogeneous ideal gas, where the binary interactions between molecules are described by the pseudo-Maxwellian kernel, a simplified model that gave rise to a variety of interesting mathematical results.\cite{Bob2}
For spatially inhomogeneous models, more sophisticated techniques are needed, and much less is known~\cite{Villani,Dolbeault}.

\section{The legacy of Maxwell molecules}\label{sec:Maxwell}
% 
% In the kinetic theory of rarefied gas, the evolution of the phase space density $f$ of the molecules of an ideal gas, $f=f(\bx,\bv,t)$ depends on the position and velocity $\bx,\bv \in \R^3$ and on time $t \in \R^+$, as described by the Boltzmann equation\cite{Cer,CIP}
% 
In the kinetic theory of rarefied gas, the evolution of the molecules composing an ideal gas is captured through a distribution function $f=f(\bx,\bv,t)$, \textit{i.e.}, a density on the phase space, which describes the proportion of molecules at position $\bx$, with velocity $\bv$ at time $t \in \R^+$.
The density $f$, often referred to as the \emph{distribution function} in kinetics, evolves according to the following Boltzmann equation\cite{Cer,CIP}
\begin{equation}
    \label{Boltz}
    {\partial\over \partial t}f(\bx,\bv,t) = -\bv\cdot\nabla_\bx f(\bx,\bv,t) + Q(f(\bx,\bv,t),f(\bx,\bv,t)).
\end{equation}
This equation contains two operators, accounting for the two phenomena affecting the density $f$ throughout time.
The term $$-\bv\cdot\nabla_\bx f(\bx,\bv,t)$$ represents the effects of {\it streaming}; that is, the free motion
\[\
    \bx_0 \mapsto \bx_0 + (t-t_0)\bv_0\qquad \bv_0 \mapsto \bv_0
\]
of molecules between two collisions. 
The term $Q(f(\bx,\bv,t),f(\bx,\bv,t))$ represents the effects of interactions among molecules, and accounts for all kinematically possible binary collisions.
% 
% , \textit{i.e.} those that conserve both momentum and energy. 
% 
In what follows, we then consider only binary and conservative collisions: only two molecules at time can interact with each other and the total energy of the system remains constant overtime.
Moreover, unless we specify otherwise, we restrict our study to the homogeneous case, so that $f$ does not depend on $\bx$.
Under these premises, the functional $Q(f,f)$ ---which acts on velocity variables only--- can be expressed as
\begin{equation}
    \label{colli}
    Q\bigl(f,f\bigr)(\bv) = \int_{{\R}^3\times{S^2}} \mathcal{K}\left(\nu, {|\bv-\bw|} \right) \big[f( \bv^*)f( \bw^*)-f( \bv)f( \bw)  \big]\, d \bw d{\bf n},
\end{equation}
where
\begin{enumerate}[label=(\roman*)]
    \item the pairs $(\bv,\bw)$ and $(\bv^*,\bw^*)$ represent the pre and post-collisional velocities in a binary collision between two elastically colliding molecules, respectively. They are related by the following identities
    \begin{equation}
        \label{col}
        \bv^* = {1\over 2}(\bv+\bw+|\bv-\bw|{\bf n}) ,\qquad \bw^* = {1\over 2}(\bv+\bw-|\bv-\bw|{\bf n}),
    \end{equation}
    where ${\bf n}$ is a unit vector and $|\cdot|$ is the Euclidean norm, that is $|\bx|=\sqrt{\sum_{i=1}^nx_i^2}$.
    % \[
    %     |\bx|=\sqrt{\sum_{i=1}^nx_i^2}.
    % \]
    % 
    \item $d{\bf n}$ denotes the {\it normalized} surface measure on the unit sphere $S^2$.
    \item $\mathcal{K}(\nu, |\bv-\bw|)$, where $\nu = (\bv - \bw)\cdot{\bf n}/|\bv - \bw|$, denotes the interaction potential, or \textit{kernel}.
\end{enumerate}

The conservation properties of the collision integral \fer{colli} are easily derived using Maxwell's weak form for the collisional integral\cite{Bob1}
\begin{equations}\label{weak}
  \langle \phi(\bv), Q\bigl(f,f\bigr)(\bv)\rangle &=  \frac 14 \int_{{\R}^3\times\R^3\times{S^2}} \mathcal{K}\left(\nu, {|\bv-\bw|} \right) \cdot \\
 &\quad\quad\quad \cdot \left(\phi(\bv^*)+ \phi(\bw^*) - \phi(\bv)-\phi(\bw)\right)f( \bv)f( \bw)\, d\bv d \bw d{\bf n},
\end{equations}
where $\phi$ is a generic continuous and bounded function.
By construction, the molecules interactions are elastic, thus the energy is conserved.
In particular, we have that $\bv^* +\bw^* = \bv +\bw$ and $|\bv^*|^2 +|\bw^*|^2 = |\bv|^2 +|\bw|^2$.
Therefore, the choice $\phi(\bv) = 1, \bv, |\bv|^2$ implies
\[
\langle \phi(\bv), Q\bigl(f,f\bigr)(\bv)\rangle = 0,
\]
corresponding to the conservations of mass, momentum, and energy, which turns out to be the only \emph{collision invariants}. 
Further, owing to the properties of the logarithm, it is immediate to verify that the choice $\phi(\bv) = \log f(\bv)$ in \fer{weak} leads to the inequality
\[
\langle \log f(\bv), Q\bigl(f,f\bigr)(\bv)\rangle \le 0.
\]
This property is the basis of Boltzmann's $H$-theorem, that is classically referred to the spatially homogeneous case under the additional condition that the streaming term is not present.
Indeed, taking into account the mass conservation property $\langle 1,Q(f,f)(\bv)\rangle = 0$, it is immediate to conclude that 
\begin{equation}
    \label{H-dec}
    \frac d{dt} \int_{\R^3} f(\bv,t)\log f(\bv,t) \, d\bv = \langle 1+ \log f(\bv,t), Q\bigl(f,f\bigr)(\bv)\rangle \le 0.
\end{equation}
Consequently, the functional,
known as entropy:
\begin{equation}
    \label{entropy}
    H(f)(t) = -\int_{\erre^n} f(\bv, t) \log f(\bv, t) \, d\bv  
\end{equation} 
increases with respect to $t$ until $\log f(\bv)$ belongs to the space of collision invariant functions. 
The equilibrium condition
\[
    \log f(\bv) = \alpha + {\bbeta} \cdot \bv + \gamma |\bv|^2, \quad \alpha,{\bbeta}, \gamma {\rm \,\, constants}
\]
combined with the conservation of mass, momentum, and energy, allow to determine unambiguously the equilibrium distribution, also known as Maxwellian equilibrium distribution
\begin{equation}
    \label{gau}
    M^f(\bv) = M_{\rho, \bu,T}(\bv) := \frac \rho{(2\pi T)^{3/2}}\exp\left\{-\frac{|\bv -\bu|^2}{2T}\right\},
\end{equation}
where $\rho, \bu, T$ are the moments of the initial distribution, that is
\begin{equation}
    \label{eq:moments_fixed}
    \rho = \int_{\R^3}f(\bv,0) d \bv, \quad \bu = \frac 1\rho\int_{\R^3}\bv f(\bv,0) d \bv, \quad
T = \frac 1{3\rho}\int_{\R^3}|\bv-\bu|^2 f(\bv,0) d \bv.
\end{equation}
% 
% Given a distribution $f$ with moments $\rho, \bu, T$, we define its Maxwellian equilibrium distribution by $M^f$.
% 

% 
Notice that, in view of the conservation of mass, if the initial condition $f(\cdot,0)$ is a probability density, \textit{i.e.}  $\int_{\erre^n} f(\bv,0)\,d\bv =1$, then the solution to the Boltzmann's equation remains a probability density for all subsequent times. 
Consequently, results about convergence towards equilibrium in terms of \emph{difference} between solutions immediately translate in results about \emph{divergence between probability distributions.} 
For this reason, in what follows, we assume that all the distributions under study are probability distributions.

\subsection{The relative entropy} 

Given two absolutely continuous probability distributions over $\erre^n$ whose densities are $f(\bv)$ and $g(\bv)$, their relative entropy is defined as
\begin{equation}
    \label{rel-ent}
    H(f|g)(t) =  \int_{\R^n} f(\bv) \log \frac{f(\bv)}{g(\bv)}\, d\bv \ge 0.
\end{equation}
Historically, given the solution $f(\bv,t)$ of the spatially homogeneous Boltzmann equation with a general kernel $\mathcal{K}(\nu, \cdot)$, Boltzmann's $H$-Theorem has always been regarded as an essential tool for determining the rate of (monotonic) convergence to zero of the relative entropy
\[
    H(f(t)|M^f)=  \int_{\R^3} f(\bv,t) \log \frac{f(\bv,t)}{M^f(\bv)}\, d\bv.
\]
The decrease in time of the relative entropy easily follows from \fer{H-dec} combined with the conservation of mass, momentum, and energy of the solution to the spatially homogenous Boltzmann equation. 
Indeed, if  $M^f$ is defined by \fer{gau} and $H(f)$ by \fer{entropy}, the conservation laws imply the equality
\[
    H(f(t)|M^f) =H(M^f)- H(f(t)).
\]
Quantifying the rate of convergence to equilibrium is not only an interesting mathematical question to be solved, but also has profound implications for applied mathematics. 
The main reason is that, while in the spatially {extended} setting the Boltzmann equation \fer{Boltz} provides a very accurate description of reality by modelling the particle  interactions in the position and momentum phase space, it is prohibitively expensive for numerical simulations of real three-dimensional problems, and one frequently resorts to fluid dynamical equations, which only uses position and time as independent variables. 
A fine kinetic description is crucial to understand the range of validity of any simplification.
This can be easily seen resorting to the classical splitting method, widely used in numerical simulations.\cite{Blanes} For the Boltzmann equation \fer{Boltz} the splitting method consists in applying sequentially on a small interval of time $\Delta t$ the collision operator
\begin{equation}
    \label{col1}
    \frac\partial{\partial t}f(\bx,\bv,t) = Q(f(\bx,\bv,t),f(\bx,\bv,t)),
\end{equation}
and the streaming operator
\begin{equation}
    \label{stre}
    \frac\partial{\partial t}f(\bx,\bv,t) = -\bv\cdot\nabla_\bx f(\bx,\bv,t).
\end{equation}
If, in a small interval of time, the solution to equation \fer{col1} relaxes to its local equilibrium, namely the Maxwellian function \fer{gau} with moments $\rho, \bu$, and $T$, as in \fer{eq:moments_fixed}, we can use the equilibrium to evaluate the evolution in time of its moments, thus reducing the complexity of numerical simulations.
We remark that the study of the rate of convergence toward equilibrium of the spatially homogeneous Boltzmann equation has achieved essential progress in the context of Maxwellian pseudo-molecules, \textit{i.e.} a simplified model in which the interaction kernel $\mathcal{K}$ depends only on the angle $\nu$ and affects only molecules that are sufficiently close.\cite{Bob1}
In this case, the collision term takes the simplified form
\begin{equation}
    \label{max}
    Q_M\bigl(f,f\bigr)(\bv) = \int_{{\R}^3\times{S^2}} \mathcal{K}\left(\nu\right) \left[ f( \bv^*)f( \bw^*)- f( \bv)f( \bw)\right] \, d \bw d{\bf n},
\end{equation}
where we recall that $\nu=\frac{(\bv-\bw)\cdot {\bf n}}{|\bv-\bw|}$ and
\[
    \int_{{S^2}} \mathcal{K}(\nu)\, d{\bf n} = 1.
\]
Notice that, unlike~\fer{colli}, in \fer{max} the interaction potential $\mathcal{K}(\cdot)$ does not depend on the relative velocity $\bv -\bw$.
The spatially homogeneous Boltzmann equation for Maxwell pseudo-molecules,\cite{Cer,CIP} which reads as
\begin{equation}
    \label{BE}
    {\partial\over \partial t}f(\bv,t) = Q_M\bigl(f,f\bigr)(\bv,t)
\end{equation}
has a number of interesting properties. 
For example, it can be simplified resorting to the Fourier transform:\cite{Bob1} 
\begin{equation}\label{eq:Fourier}
    \f(\bxi,t) = \int_{\R^3} f(\bv,t) e^{-i\bxi\cdot\bv}\, d\bv.
\end{equation}
Indeed, owing to the properties of the Fourier transform, Boltzmann's equation \fer{BE} can be rewritten as follows
\begin{equation}
    \label{BEF}
    {\partial\over \partial t}\f(\bxi,t) = \int_{{S^2}} \mathcal{K}\left(\nu\right) \left[\f\left( \frac{\bxi +|\bxi|{\bf n}}2\right)\f\left( \frac{\bxi -|\bxi|{\bf n}}2\right)- \f( \bxi)\f(0)\right] \,  d{\bf n}
\end{equation}
in which integration is reduced in dimension. 
% \textcolor{purple}{[GB: is the variable $\nu$ or $\mathbf n$? Am I missing something?]}

% Equations \fer{BE} and \fer{BEF} are among the most intensively studied models, as their collision frequency $B(\cdot)$ is independent of the relative velocity of the colliding pair. 

%The investigation of the spatially homogenous Boltzmann equation and of its simplified models made possible to achieve essential progresses on the relaxation process. 

\subsection{Approaching the problem through a different divergence, the relative Fisher information}

The relative entropy is not the unique tool used to measure the divergence between a solution of a kinetic model and its equilibrium state.
Given $f= f(\bv)$ a probability density on $\R^n$, Fisher’s quantity of information associated to $f$ is defined as the (possibly infinite) non-negative number
\begin{equation}
    \label{Fisher}
    I(f) = \int_{\R^n}\frac{|\nabla f(\bv)|^2}{f(\bv)} \, d\bv = 4\int_{\R^n}|\nabla \sqrt{f(\bv)}|^2 \, d\bv.
\end{equation}
This formula defines a convex, isotropic functional $I$, which was first used by Fisher in Ref. \cite{Fish} for statistical purposes, and plays a fundamental role in information theory and in other contexts.\cite{Vil3}
% 

% 
% As already mentioned, 
% 
Fisher information was firstly introduced in kinetic theory by McKean Jr. in Ref. \cite{McK}, to study the convergence to equilibrium of the solution to the one-dimensional Kac caricature of the Boltzmann equation with a Maxwellian kernel.
%
% o study the convergence to equilibrium for the solution to the so-called one-dimensional Kac caricature of the Boltzmann equation with a Maxwellian kernel. 
% 
The key observation by McKean was that the functional $I$, like the classical Boltzmann functional $H$, is non-increasing in time along solutions of Kac’s model. 
This monotonicity property was extended in Ref. \cite{BT} to the two-dimensional Boltzmann equation for Maxwellian molecules, and finally generalized to higher dimensions of the velocity space by Villani in Ref. \cite{Vil1} (cfr. also Ref. \cite{MT} for a simplified proof based on the Fourier expression \fer{BEF}). 
It is interesting to remark that, if $M^f(\bv)$ is the Maxwellian function defined in \fer{max}, its Fisher information only depends on its principal moments, in particular
\[
    I(M^f) = \frac{3\rho}{T}.
\]
Moreover, if we define the relative Fisher information of two probability densities $f$ and $g$ through the formula: 
\begin{equation} 
    \label{rel-Fish}
    I(f|g) =\int_{\R^n}\left| \frac{\nabla f(\bv)}{f(\bv)} - \frac{\nabla g(\bv)}{g(\bv)}\right|^2 f(\bv) \, d\bv,
\end{equation}
one has, like in the case of the relative entropy, that the relative Fisher information coincides with the difference between the Fisher measures of information, that is
\[
    I(f|M^f) = I(f) -I(M^f).
\]

\subsection{Energy Distances from Kinetic Theory}

% Boltzmann type equations with a Maxwellian kernel have been recently considered also in statistical problems of interest in social problems related to economy.
% 
Boltzmann-type equations with a Maxwellian kernel have recently been applied to statistical problems of interest in socio-economic contexts.\cite{CoPaTo05}
In this new setting, one aims to follow the distribution of wealth in a multi-agent system in which the relevant interactions are binary trades with a risky component. Let $f(w,t)$ denote the probability density at time $t \ge 0$ of agents with personal wealth $w\ge0$, departing from an initial density $f_0(w)$ with a mean value fixed equal to one. Then, as in the classical case described above, the \emph{collision} integral \fer{max} takes into account the elementary (one-dimensional) interactions 
\be\label{trade}
v^* = (1-\lambda)v + \lambda w + \eta v; \quad w^* = (1-\lambda)w + \lambda v + \tilde\eta w 
\ee
where $0<\lambda<1$ is a suitable constant, while $\eta$ and $\tilde\eta$ are two independent and equi-distributed random variables of mean value equal to zero and variance $\sigma$. In Ref. \cite{CoPaTo05} it has been shown that, in a certain limit, the Boltzmann equation is well approximated by the one-dimensional Fokker--Planck equation\cite{BM} 
 \be\label{FP2c}
 \frac{\partial f(w,t)}{\partial t} = J(f)(w,t) = \frac \sigma{2}\frac{\partial^2 }
{\partial w^2}\left( w^2 f(w,t)\right) + \lambda \frac{\partial }{\partial w}\left(
(w-1) f(w,t)\right).
\ee
 The key features of equation \fer{FP2c} is  that, in presence of suitable boundary conditions at the point $w=0$, the solution preserves both mass and momentum.
Moreover, it approaches in time a unique stationary solution of unitary mass.\cite{TT1}
This stationary state is given by the inverse Gamma distribution
 \be\label{equi2}
f_\infty(w) =\frac{(\mu-1)^\mu}{\Gamma(\mu)}\frac{\exp\left(-\frac{\mu-1}{w}\right)}{w^{1+\mu}},
 \ee
  where the positive constant $\mu >1$ is given by
\[
\mu = 1 + 2 \frac{\lambda}{\sigma}.
\]
% As predicted by the observations of the Italian economist Vilfredo Pareto in Ref. \cite{Par}
% 
In particular, we observe that the equilibrium solution \fer{equi2} exhibits a power-law tail for large values of the wealth variable as predicted by the Italian economist Vilfredo Pareto.\cite{Par} 
Lastly, the rate of convergence to equilibrium of the solution to equation \fer{FP2c} was studied in Ref. \cite{ToTo}, by expressing the linear kinetic equation in the Fourier space.
For any given $r\in \R$ it was shown that  
\begin{equations}
    \int_\R |\xi|^{2r}|\hat f(&\xi,t)-\hat f_\infty(\xi)|^2 \, d\xi \le \\  &\exp\left\{\frac{2r+1}2\left[ \sigma\frac{2r+3}2 +2\right]t\right\}\int_\R |\xi|^{2r}|\hat f_0(\xi,t)-\f_\infty(\xi)|^2\, d\xi,
\end{equations}
where $\f(\xi,t)$ is the Fourier transform of the density $f(w,t)$. Choosing   {$r<-1/2$} one concludes with the {exponential convergence} of   $\f(\xi,t)$ towards $\hat f_\infty(\xi)$  in the Sobolev space $\dot{H}_{-r}$.
It is interesting to remark that the choice $r = -1$ is nothing but the well-know original concept of Energy Distance introduced by Sz\'ekely.\cite{Sze89,Sze03}

\section{Why Kinetic Theory Should Matter to AI}

% 
% As the previous discussion enlightens, 
% 
The study of the convergence to the Maxwellian equilibrium of the solution of the Boltzmann equation for Maxwell pseudo-molecules has led to identify new mathematical tools, and among them, various \emph{divergences} for probability measures tailor made to the specific structure of the collision integral in equation \fer{max}.  
However, the methodology of kinetic theory is not restricted to the study of ideal gases nor to binary interactions. 
Indeed, kinetic modelling as been used to study social sciences, leading to the development of new applied models.\cite{NPT,PT13}
% 

% 
% In many of these applications, binary interactions between molecules are replaced by linear interactions with an external background, which remains constant over time. 
% % 
% Such linear models are described by equations of the type \fer{col1}, that is
% % 
% \begin{equation}
%     \label{col-hom}
%     \frac\partial{\partial t}f(\bv,t) = Q(f(\bv,t), B(\bv)),    
% \end{equation}
% % 
% where $B(\bv)$ is the density function that describes the action of the background on the particles. 
% % 
% Within the classical field of binary collisions between molecules, if we choose $B$ to be equal to the Maxwellian equilibrium distribution $M^f(\bv)$ defined in \fer{gau}, the interaction \fer{col} characterizes the collisions of molecules with a fixed background in a Maxwellian equilibrium whose mass, momentum, and energy are fixed.
% % \textcolor{purple}{[GB: agree.]}
% % 
% In this case, Maxwell’s weak form for the collisional integral simplifies to 
% \begin{equation}
%     \label{weak-lin}
%     \langle \phi(\bv), Q\bigl(f,M^f\bigr)(\bv)\rangle \! =\! \frac 12 \int_{{\R}^3\times\R^3\times{S^2}} \!\!\!\!\mathcal{K}\left(\nu\right) 
%     \left(\phi(\bv^*) - \phi(\bv)\right)f( \bv)M^f( \bw)\, d\bv d \bw d{\bf n},
% \end{equation}
% where $M^f$ is defined as in \eqref{gau}.

% 
In this section, we aim at establishing a connection between the Neural Networks commonly used in Machine Learning and kinetic models. 
In particular, we consider the structure of the simplest neural network: a Feedforward Neural Network (FNN).

\subsection{The Feedforward Neural Network}

A Feedforward Neural Network (FNN) can be graphically represented as a sequential $T+2$ partite orientated graph, that is, a collection of $T+2$ sets of nodes $\{N_i\}_{i=1,\dots,T+2}$, also known as \textit{layers}, that are ordered sequentially.
Every edge adjacent to a node in $N_i$ either goes from a node in $N_i$ to a node in $N_{i+1}$ or goes from a node in $N_{i-1}$ to a node in $N_i$, so that any flow on the directed graph flows in one direction.
For the sake of simplicity, we assume that every layer contains the same amount of nodes and that each node in $N_i$ is connected to every node in $N_{i+1}$, so that the FNN is fully connected.
The first layer is known as the \textit{input layer}, the last layer is known as the \textit{output layer}, the other $T$ layers are called \textit{hidden layers} and are the layers through which the information fed through the input layer flows and gets processed before reaching the output layer.
In Figure \ref{fig:deep_feedforward_nn}, we give a graphical representation of the structure of a FNN.

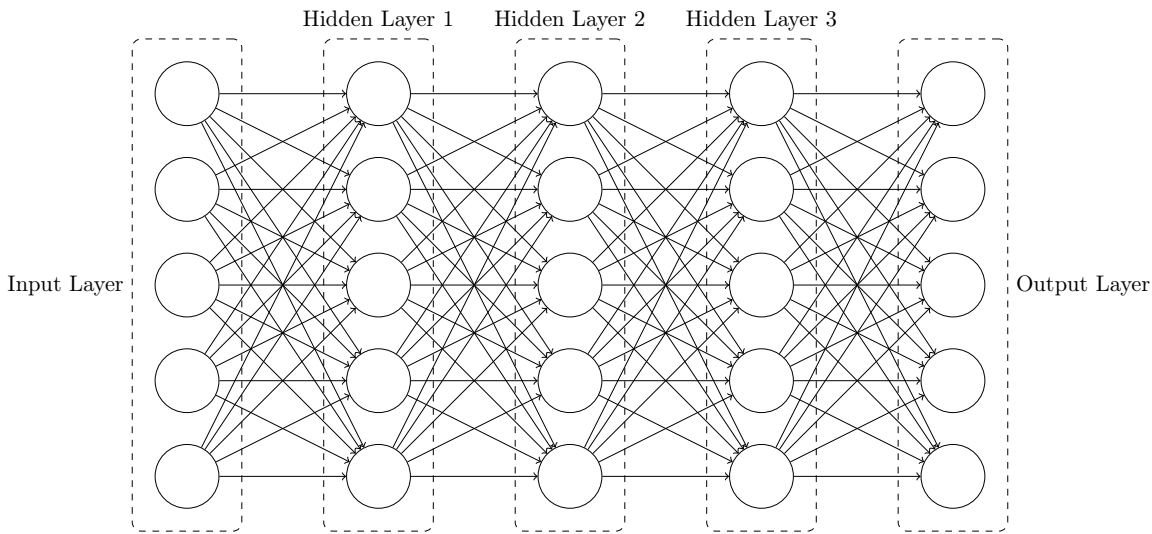
\begin{figure}[t!]
    \centering
    \resizebox{\textwidth}{!}{
    \begin{tikzpicture}[
        neuron/.style={circle, draw, minimum size=1cm},
        layer/.style={draw, dashed, rounded corners, inner sep=10pt}
    ]

        % Define nodes for each layer (5 neurons per layer)
        \foreach \i in {1,2,3,4,5} {
            \node[neuron] (I\i) at (0,3-\i*1.5) {};  % Input layer
            \node[neuron] (H1\i) at (3,3-\i*1.5) {}; % First hidden layer
            \node[neuron] (H2\i) at (6,3-\i*1.5) {}; % Second hidden layer
            \node[neuron] (H3\i) at (9,3-\i*1.5) {}; % Third hidden layer
            \node[neuron] (O\i) at (12,3-\i*1.5) {}; % Output layer
        }

        % Draw layer boxes
        \node[layer, fit={(I1) (I5)}, label={left:Input Layer}] {};
        \node[layer, fit={(H11) (H15)}, label={above:Hidden Layer 1}] {};
        \node[layer, fit={(H21) (H25)}, label={above:Hidden Layer 2}] {};
        \node[layer, fit={(H31) (H35)}, label={above:Hidden Layer 3}] {};
        \node[layer, fit={(O1) (O5)}, label={right:Output Layer}] {};

        % Draw connections between layers
        \foreach \i in {1,2,3,4,5} {
            \foreach \h in {1,2,3,4,5} {
                \draw[->] (I\i) -- (H1\h);
                \draw[->] (H1\i) -- (H2\h);
                \draw[->] (H2\i) -- (H3\h);
                \draw[->] (H3\i) -- (O\h);
            }
        }

    \end{tikzpicture}
    }
    \caption{A simple example of a fully connected Feedforward Neural Network (FNN) with five layers and five neurons per layer.}
    \label{fig:deep_feedforward_nn}
    
\end{figure}

The way in which the information is processed through the FNN is determined by the weights of the FNN and its activation function.
The weights of a Neural Network is a collection of values $\{w_e\}_{e\in E}$, where $E$ is the set containing all the orientated edges of the FNN.
We denote with $w_{i,j}^{(t)}$ the weight of the edge connecting node $i\in N_t$ to the node $j\in N_{t+1}$.
The activation function, which we denote with $\phi$, is a non-linear application that maps a real value to another real value.
Let us now assume that we are given an input for our FNN which, owing to the structure of the network, consists of $M$ real values (for example they could be the medical information of an individual or the greyscale values of a black and white image).
Fixed the weights of all the edges, the information given in input is processed via the following iterative process.
Given the $M$ real values of the $t$-th layer, namely the $M$-dimensional vector $V^{(t)}=(\alpha_1^{(t)},\alpha_2^{(t)},\dots,\alpha_{M}^{(t)})$ we have that the information at the $(t+1)$-th layer is determined by two sequential operations:
\begin{enumerate}
    \item[(O1)] \label{first:op:FNN} A linear application, that assigns a real value to the $i$-th node of the $(t+1)$-layer through the formula
    \begin{equation}
    \label{eq:layert+1}
        \alpha_{i}^{(t+1)}=\sum_{k=1}^M w_{k,i}^{(t)}\alpha_{k}^{(t)}.
    \end{equation}
    \item[(O2)] \label{second:op:FNN} The activation function $\phi$, which maps the values in \eqref{eq:layert+1} to $\phi(\alpha_{i}^{(t+1)})$.
\end{enumerate}
% 
% \textcolor{purple}{[GB: I would highlight (1) and (2) a bit more, e.g., with colour.]}

\subsection{The FNN through Boltzmann's Equation}

In many of kinetic models, binary interactions between molecules are replaced by linear interactions with an external background, which remains constant over time. 
Such linear models are described by equations of the type \fer{col1}, that is
\begin{equation}
    \label{col-hom}
    \frac\partial{\partial t}f(\bv,t) = Q(f(\bv,t), B(\bv)),    
\end{equation}
where $B(\bv)$ is the density function that describes the action of the background on the molecules. 
Within the classical field of binary collisions between molecules, if we choose $B$ to be equal to the Maxwellian equilibrium distribution $M^f(\bv)$ defined in \fer{gau}, the interaction \fer{col} characterizes the collisions of molecules with a fixed background in a Maxwellian equilibrium whose mass, momentum, and energy are fixed.
% \textcolor{purple}{[GB: agree.]}
% 
In this case, Maxwell’s weak form for the collisional integral simplifies to 
\begin{equation}
    \label{weak-lin}
    \langle \phi(\bv), Q\bigl(f,M^f\bigr)(\bv)\rangle \! =\! \frac 12 \int_{{\R}^3\times\R^3\times{S^2}} \!\!\!\!\mathcal{K}\left(\nu\right) 
    \left(\phi(\bv^*) - \phi(\bv)\right)f( \bv)M^f( \bw)\, d\bv d \bw d{\bf n},
\end{equation}
where $M^f$ is defined as in \eqref{gau}.

% 
% We are now ready to relate the structure of the FNN to kinetic theory.
% 
We now show that it is possible to interpret an FNN as one of these kinetic models.
For the sake of simplicity, let us restrict our study to the case in which each layer contains the same amount of nodes, namely $M$, and the coefficients in \fer{eq:layert+1} do not depend on $t$, but can in general be distributed according to some probability law. 
Within this framework, operation (O1) can be reinterpreted as an interaction between a molecule with velocity $\bV=(\alpha_1,\dots,\alpha_M)$, with a background $B$, that returns a velocity $\bV^* = A^T\bV$, where $A$ is the $M\times M$ matrix with non negative weights $w_{i,j}$ that characterizes the action of the background.
Likewise, the role of the activation function $\phi$ described by operation (O2) can be interpreted as the action of a second background which, given the velocity $\bV$ returns a velocity $\bV^{**} = (\phi(\alpha_1),\dots, \phi(\alpha_M))$.
Consequently, postulating the validity of molecular chaos,\footnote{\textit{i.e.} the distribution of the molecules is independent if the number of molecules $n$ is large enough} the evolution of the FNN's output can be described by a linear Boltzmann equation of type \fer{col-hom}, where now the collision term is the sum of two different contributions. 
Assuming the description of Maxwell-type pseudo-molecules, which in this case correspond to fix a constant interaction potential, we obtain that the distribution $f(\bV,t)$ of molecules obeying  the rules of the FNN network evolve according to
\begin{equation}
    \label{col-FNN}
    \frac\partial{\partial t}f(\bV,t) = \frac 1{\tau_1} Q_1 (f(\bV,t), B)+ \frac 1{\tau_2} Q_2 (f(\bV,t),f(\bV,t))
\end{equation}
where $\tau_i, \, i=1,2$, are suitable relaxation times, so that the Maxwell’s weak form for the collisional integral simplifies to 
\begin{equation}
    \label{lin-1}
    \langle \phi(\bV), Q_1\bigl(f,M\bigr)(\bV)\rangle =  \frac 1{2\tau_1} \int_{{\R}^M\times\R_+} \left(\phi(\bV^*) - \phi(\bV)\right)f( \bV)M( w)\, d\bV d w 
\end{equation}
and
\begin{equation}
    \label{lin-2}
    \langle \phi(\bV), Q_2\bigl(f,B\bigr)(\bV)\rangle = \frac 1{2\tau_2} \int_{{\R}^M} \left(\phi(\bV^{**}) - \phi(\bV)\right)f( \bV)\, d\bV.
\end{equation}
The integral on the right-hand side of equation \fer{lin-2} measures the variation of the observable quantity $\phi$ following the variation of the quantity $V$ under consideration, which is determined by the action of the two backgrounds.
% 
% Notice that \eqref{lin-1,lin-2} describe the average change of the observable $\phi$ after the interaction of the molecules, which is represented by the quantity $\phi(\bV^{**}) - \phi(\bV)$, times the probability of the interaction happening.}

% 

\subsection{The FNN through Generalized Maxwell Models}

To reinforce our argument, we present a different way of looking at the kinetic description of a FNN network based on a generalized Maxwell model in which multiple interactions can happen at the same time.
These models have been fruitfully used in Ref. \cite{BCG}, where they have been deployed to describe the evolution of wealth distribution within a fixed population.
% 
% These models have been fruitfully used in Ref. \cite{BCG}, where generalized Maxwell models of the Boltzmann equation in presence of multiple interactions have been adopted to describe the evolution of the wealth distribution within a fixed population.
% 
% An interesting related application  can be found in Ref. \cite{BCG}, where generalized Maxwell models of the Boltzmann equation in presence of multiple interactions were introduced and studied, having as main application the economic context. 
% 
Specifically, in Ref. \cite{BCG} the microscopic interaction between $M\ge 1$ particles is interpreted as an interaction in a community of agents participating in economical trades.
Denoted with $\bV= (v_1,v_2,\dots, v_M)$ the $M$-dimensional vector of the non-negative wealths of the agents, the post-interaction state $\bV^\ast= (v_1^*,v_2^*,\dots, v_M^*)$ gives the new wealths after a single economic trade. 
The $M$-particle interaction considered in Ref. \cite{BCG} is a linear transformation of $\bV$ given by
\begin{equation}
    \label{n-th}
    v_i^\ast=av_i+b \sum_{j=1}^{M}v_j \qquad i=1,\,\dots,\,M,
\end{equation}
where the parameters $a,\,b$ are distributed according to a probability distribution with a suitable number of bounded moments. 
Let us now denote by $f_i=f(v_i,\,t)$ the density of particles with state $v_i\in\R$ at time $t\geq 0$. 
According to \fer{n-th}, each element $v_i$ of the vector $\bV$ is subject to the same variation in terms of the other $M-1$ elements.
Consequently, if we start with a common initial distribution $f_i(v_i, t=0) = f(v_i, t=0)$, $i = 1,,\dots, M$, at any subsequent time $ t >0$ one has $f_i(v_i, t) = f(v_i, t)$.
Postulating the validity of molecular chaos, the evolution of any \textit{observable quantity} $\varphi$
% , \textit{i.e.} any quantity which may be expressed as a continuous and bounded function of $v$, 
is given by the following Boltzmann-type equation
\begin{equation}
    \frac{d}{dt}\int_{\R_+}\varphi(v)f(v,\,t)\,dv=
        \frac{1}{\tau M} \int_{\R_+^M}\sum_{i=1}^{M}\left\langle{\varphi(v_i^\ast)-\varphi(v_i)}\right\rangle{\prod_{j=1}^{n}f(v_j,\,t)}\,dv_1\,\dots\,dv_M,
    \label{eq:kine-Gamba}
\end{equation}
where $\tau$ denotes a relaxation time and $\langle{\cdot}\rangle$ is the average with respect to the distributions of the random parameters $a,\,b$ in~\fer{n-th}. 
The collision integral depends on a constant \textit{collision kernel},  that corresponds to consider \textit{Maxwell-type interactions}. 
The right-hand side of~\fer{eq:kine-Gamba} takes into account the whole set of microscopic states, and consequently it depends on the $n$-product of the density functions $f(v_1,\,t)\cdot\ldots\cdot f(v_M,\,t)$. 
Thus, if $M>1$ the evolution of $f$ obeys a highly non-linear Boltzmann-type equation. 
The interesting point remarked in Ref. \cite{BCG} is that, in this case, a considerable simplification occurs in presence of a large number of participants $M$, resulting in a linearized version of equation \fer{eq:kine-Gamba}.
These findings are an interesting example of the way kinetic theory of rarefied gases can be fruitfully employed to study multi-agents problems in economics.
Related models based on multiple interaction of type \fer{n-th} have been subsequently considered in Ref. \cite{TTZ} in connection with the study of a jackpot game with $M$ gamblers. 
Going back to the FNN, it becomes evident that operation \eqref{eq:layert+1} described in (O1) can be interpreted as a generalization of the $n$-th interaction \fer{n-th} discussed in Ref.~\cite{BCG}, with the linear transformation representing the core of the information processing.
In the case described by the FNN, the interaction between $n$-agents with \emph{velocities} $(\alpha_1,\alpha_2,\dots,\alpha_{M})$ changes the velocities into $(\alpha_1^*,\alpha_2^*,\dots,\alpha_{M}^*)$ according to the linear formula \eqref{eq:layert+1}. 
Afterwards, in operation (O2), the velocities $\alpha_i$ are modified into $\alpha^{**}_i = \phi(\alpha_i)$, $i =1,2,\dots, M$.
Let us now denote by $f_i=f(\alpha_i,\,t)$ the distribution  of the values of the node $i\in N$ at layer $t$. 
According to \eqref{eq:layert+1}, in presence of a general set of weighs, the velocities $\alpha_i$ are subject to different variations in terms of the other $M-1$ elements. %\textcolor{purple}{[GB: should we make more explicit that the previous case corresponds to a special example of the current one with $w_{k,i} = a \, \delta^{i,k} + b (1-\delta^{i,k})$?]}
Consequently, at difference with \fer{n-th}, even if we start with a common initial distribution $f_i(\alpha_i, t=0) = f(\alpha_i, t=0)$, $i = 1,,\dots, M$, at any subsequent time $ t >0$ the densities $f_i(\alpha_i, t)$ follow different evolutions. 
Hence, the (unique) Boltzmann equation \fer{eq:kine-Gamba} is replaced by a system of $M$ Boltzmann equations of the form
\begin{equations}
    \frac{d}{dt}\int_{\R_+}\varphi(\alpha_i)f_i(\alpha_1,\,t)\,dv=& \frac{1}{\tau_1 }\int_\R (\psi(\alpha^{**})- \psi)(\alpha))f_1(\alpha_i,t)\, d\alpha_i\\
      & \;\; +\frac{1}{\tau_2 } \int_{\R_+^M}\left({\varphi(\alpha_i^\ast)-\varphi(\alpha_i)}\right){\prod_{j=1}^{n}f_j(\alpha_j,\,t)}\,d\alpha_1\,\dots\,d\alpha_M,
    \label{eq:kine-system}
\end{equations}
with $i=1,2,\dots,M$. 

In particular, once a kinetic model is formalized, be it in the form \fer{col-FNN} or \fer{eq:kine-system}, the properties of the probability density solution of the model, and eventually of the associated steady state solution, can be achieved by resorting to the mathematical techniques developed for the classical Boltzmann equation.

\section{Towards spatially inhomogeneous models}
In the previous section of the paper, we have been considering a few examples where kinetic equations ---and the divergences adapted to quantify their convergence to equilibrium--- apply to models from AI and statistics. 
Our analysis has been carried out in the \emph{physical space} $\mathbb R^n$.

Actually, kinetic theory provides us with a natural hierarchy between a \emph{position} and \emph{velocity} variables, denoted with $\bx$ and $\bv$, respectively, as in \eqref{Boltz}. Let us consider, henceforth, the phase space of positions and velocities $\Gamma = \{(\bx,\bv) \, :\, \bx,\bv \in \mathbb R^n\}$.
Such a structure could be exploited, beyond its physical relevance, in applications within the scope of AI. Let us briefly review two of them.
\subsection{From exponential to ballistic along the \texorpdfstring{$\chi^2$-distance}{}}
Let $\mathbf{X} \sim \mu=\mu(d\bx)$ be a random variable on the physical space $\mathbb R^n$. In order for a NN to \emph{learn} the distribution $\mu$, a dataset of samples of $\mu$ is needed. However, simulating the measure $\mu$ directly is practically unfeasible in most cases. 
Then, an idea\cite{Wibisono2018,Welling2011} is to construct a time-evolution\footnote{which, in practice, is time-discrete, but for presentation purposes will be written as time-continuous} process $(\mathbf{X}_t)_{t\geq0}$, starting at a given (and easily accessible) variable $\mathbf X_0 \sim \mu_0$, and such that the law $\mu_t$ of $\mathbf X_t$ is such that $\mu_t \to \mu$, as $t\to\infty$. Then, we sample from $\mathbf{X}_t$, which is accessible and, for $t$ large enough, we will have a reasonable approximation of the inaccessible target $\mathbf{X}$.  

One classical example, if $\mu = \mathrm{e}^{-\phi(\bx)} \, d\bx$, for a smooth potential $\phi$, is the \emph{overdamped Langevin dynamics}, given by the equation
\begin{equation}\label{langevin}
\begin{cases}
    d \mathbf X_t = -\nabla_x \phi(\mathbf X_t) + \sqrt{2} dW_t,  \\
    \mathbf{X}_{t=0} = \mathbf{X}_0.
\end{cases}
\end{equation}
where $W_t$ is a Brownian motion over $\mathbb R^n$. Implementations of \eqref{langevin} have been studied in the recent literature, also in connection with stochastic gradient-descent methods. 
The law $\mu_t=f(\bx,t) \, d\bx$ of $\mathbf{X}_t$ obeys the Fokker--Planck equation
\begin{equation}\label{OU}
    \partial_t f = \nabla_x \cdot( \nabla_x f + \phi(\bx) f).
\end{equation}
If $\phi$ is strictly convex, i.e., $\nabla^2_{xx} \phi \geq \alpha \Id$, for some $\alpha>0$, it is well-known that\cite{bakryemery}  
\[
H(\mu_t|\mu) \leq \mathrm{e}^{- 2\alpha t} \, \mathrm{H}(\mu_0|\mu).
\]
If $\phi$ is not strictly convex, but still it is such that $\mu$ satisfies a Poincar\'e inequality\footnote{very roughly corresponding to $\phi(\bx) \gtrsim \lambda_\mu \, |\bx|$ for large values of $|\bx|$} with constant $\lambda_\mu$, then we can control the $\chi^2$-divergence, given by $\chi^2(f|g) = \int_{\erre^n} |f(\bx)/g(\bx) -1|^2 \, g^2(\bx) \, d\bx$, as follows 
\begin{equation}\label{chi2}
\chi^2(\mu_t|\mu) \leq \mathrm{e}^{-2 \lambda_\mu t} \, \chi^2(\mu_0|\mu).
\end{equation}
When $\lambda_\mu << 1$, the convergence property of the sampling method deteriorates. 

Here is where kinetic theory comes to help\cite{Lei2016}. Let us artificially introduce the velocity variable $\bv$, and let us construct a process $(\mathbf{X}_t,\mathbf{V}_t)_{t\geq 0}$ on the phase space, such that the law $\Theta_t(\bx,\bv)$ of $(\mathbf{X}_t,\mathbf{V}_t)$ converges, e.g., to $\Theta := \mathrm{e}^{-\phi(\bx) - \frac{1}{2}|\bv|^2} \, \,d\bx\,d\bv.$

A typical example is the \emph{underdamped Langevin dynamics}
\begin{equation}\label{ulangevin}
\begin{cases}
    d\mathbf{X}_t = \mathbf{V}_t, \\
    d\mathbf V_t = - \mathbf V_t - \nabla_x \phi(\mathbf V_t) + \sqrt{2}\,dW_t, \\
    (\mathbf{X}_t,\mathbf{V}_t)_{t=0} = (\mathbf X_0,\mathbf V_0),
\end{cases}    
\end{equation}
where $W_t$ is a standard Brownian motion in velocity variables only, and $(\mathbf X_0,\mathbf V_0)$ is an accessible initial measure. 
The law $\Theta_t(\bx,\bv)$ of $ (\mathbf X_t,\mathbf V_t)$ satisfies the kinetic Fokker--Planck equation
\begin{equation}
    \label{kfp}
    \partial_t \Theta_t + \bv \cdot \nabla_x \Theta_t - \nabla_x \phi \cdot \nabla_v \Theta_t = \nabla_v \cdot (\nabla_v \Theta_t + \bv\Theta_t),
\end{equation}
which can be thought of as a simplification of \eqref{Boltz}, where the collision term is replaced by interaction with a random (but fixed) background force. 
Studying \eqref{kfp} is challenging, since diffusion acts only on velocity variables, which are linked to the position variables via the Hamiltonian transport operator $\bv\cdot\nabla_x - \nabla_x \phi \cdot \nabla_v$. The hierarchy between $\bx,\bv$, and the interplay between transport and diffusion require sophisticated analytic tools in order to be dealt with.
However, by means of the \emph{hypocoercivity} theory\cite{Villani}, in particular the $\mathrm{L}^2$-framework of J.~Dolbeault, C.~Mouhot, and C.~Schmeiser\cite{Dolbeault}, it holds true that
\begin{equation}\label{chi3}
\chi^2(\Theta_t|\Theta) \leq C\, \mathrm{e}^{-2 \lambda t} \, \chi^2(\Theta_0|\Theta),
\end{equation}
for a pre-factor $C>1$, and a certain rate $\lambda>0$. Recently\cite{Cao,Brigati,BrigatiLifts,BrigatiPoincare}, constructive estimates for $C$ and $\lambda$ have been established. In particular, indicating again with $\lambda_\mu$ the Poincar\'e constant of the target measure $\mu$, it holds true that 
\[
\chi^2(\Theta_t|\Theta) \lesssim\, \mathrm{e}^{-2 \sqrt{\lambda_\mu} t} \, \chi^2(\Theta_0|\Theta),
\]
up to an explicit, dimension-free constant. 
The rate is then improved up to the \emph{ballistic} regime $\sqrt{\lambda_\mu}$, which is actually sharp.\cite{Eberle} To simulate the target $\mu$, we sample from the space-marginal of $\Theta_t$ for $t$ sufficiently large.

\subsection{Smooth interpolation: Wasser-splines, \emph{minimal acceleration}}
In his breakthrough-paper of 1991, Y.~Brenier\cite{brenier1991polar} shows the existence of optimal transport maps for the Wasserstein distance. In other words, the optimal plan $\pi$ in \eqref{Was-p} should be looked for among those which \emph{maximize correlation between the marginals}. Whenever the first marginal $\mu$ is regular enough, e.g.~$\mu(d\bx) =f(\bx)\,d\bx$,  the optimal transport plan between $\mu$ and $\nu$ is unique, and concentrates on the graph of a map $M$, such that $(M)_\# \mu = \nu$. Then, going back to the original ideas of G.~Monge\cite{monge1781memoire}, it is very natural to introduce the interpolating curve $(\bar\mu_t)_t$ given by 
\begin{equation}\label{bb}
\bar\mu_t = ((1 - t) \Id + t M)_\# \mu, \qquad t \in [0,1],
\end{equation}
which is a \emph{dynamical} way of transporting $\mu$ to $\nu$. In addition, the curve $(\bar\mu_t)_t$ is optimal ---among all continuous ways to move $\mu$ to $\nu$--- in the sense of the total action\cite{Benamou} 
\[
\int_0^1 \int_{\mathbb R^n} |V_t|^2 \, d\mu_t \, dt, \qquad \partial_t \mu_t + \nabla_x \cdot (V_t \mu_t) = 0, \quad \mu_0=\mu, \quad \mu_1 = \nu. 
\]
Finding a suitable \emph{intermediate version} between two probability measures $\mu$ and $\nu$ is useful in a number of practical situations, ranging from the treatment of images and datasets, to biology, and  urban planning.\cite{peyre2019computational,santambrogio2015optimal} 
In addition, calculating the optimal map $M$ is practically feasible, thanks to the recent progress in computational optimal transport.\cite{peyre2019computational} If $M$ is known, then, the linear interpolation $\bar\mu_t$ is simply given by following straight \emph{rays} between each pair $(\bx,M(\bx))$, for $\bx$ in the support of $\mu$.

Even if the Monge--Mather shortening principle, and the Monge--Ámpere equation\cite{villani2009optimal} guarantee some regularity properties for $\bar\mu_t$, the Wasserstein optimal interpolation is not always the most desirable choice. 
\begin{enumerate}
    \item In image interpolation, or generation ---as it is the case for the model GOTEX for synthetized textures\cite{Houdard}, a smoother curve between $\mu$ and $\nu$ will provide higher-regularity objects.\cite{BenamouGallouet} 
    \item In optimal steering of vehicles or cellular differentiation, it is sometimes more natural to follow the path which minimises the consumption of fuel/nutrient, rather than the shortest one. 
\end{enumerate}
Let us review an alternative interpolation method, based on kinetic theory and optimal control\cite{Chen,Chewi,BrigatiOtikin,Einav,Elamvazhuthi,Justiniano}, which is adaptable to both issues listed above.

Let $\mu=\mu(d\bx),\,\nu=\nu(d\bx)$ be two measures on the position space $\mathbb R^n$. Let us complete $\mu,\nu$ to two measures $\Theta, \Xi$ over the phase space $\Gamma \ni (\bx,\bv)$, under the condition that 
\[
\mathbb E_v(\Theta) = \mu, \qquad \mathbb E_v(\Xi) =\nu.
\]
The velocity marginals of $\Theta, \Xi$ may come from a measurement ---as in physics--- or are assigned artificially.\cite{Chen}
Then, let us consider the following minimisation problem
\[
\mathsf d_T^2(\Theta,\Xi) = \inf_{\pi \in \Pi(\Theta,\Xi)} \, {\int_{\Gamma^2} 12 }\left| \frac{\by-\bx}{T} - \frac{\bv+\bw}{2} \right|^2 + |\bv-\bw|^2 \, d\pi((\bx,\bv),(\by,\bw)), \qquad T>0.
\]
An optimiser $\pi$ is always admitted, but, whenever $\Theta(d\bx,d\bv) = f(\bx,\bv) \, d\bx \,d\bv$, then there exists an optimal plan $\pi$ in the form $(\Id,M_T)_\# \Theta$, for a Monge map $M_T: \Gamma \to \Gamma$. Let us indicate $(\by,\bw) := M_T(\bx,\bv)$. To interpolate on the phase space, we introduce 
\[
\begin{aligned}
&M_t(\bx,\bv) = (\alpha_{x,v}(t),\alpha'_{x,v}(t)), \qquad t \in [0,T], \\
&\alpha_{x,v}(t) = \left( \frac{\bv+\bw}{T^2}-2\frac{\by-\bx}{T^3} \right)t^3 + \left(3\frac{\by-\bx}{T^2}-\frac{2\bv+\bw}{T}\right)t^2 + \bv t+\bx
\end{aligned}
\]
The curve $\alpha_{x,v}$ is called a \emph{cubic spline} between $(\bx,\bv)$ and $(\by,\bw)$, and it has the property of minimising the acceleration along paths connecting $(\bx,\bv)$ with $(\by,\bw)$: 
\[
\begin{aligned}
&12 \left| \frac{\by-\bx}{T} - \frac{\bv+\bw}{2} \right|^2 +|\bv-\bw|^2 = T \,\int_0^T |\alpha''_{x,v}(t)|^2 \, dt \\
&= \inf_{\alpha} \, T \int_0^T |\alpha''(t)|^2 \, dt, \quad \alpha: \, [0,T] \to \Gamma, \quad (\alpha,\alpha')(0) = (\bx,\bv), \quad (\alpha,\alpha')(T) = (\by,\bw).
\end{aligned}
\]
Then, the functional $\mathsf d_T$ is named the \emph{minimal acceleration discrepancy}\cite{BrigatiOtikin}, as it is given precisely by the acceleration of optimal splines between $\Theta$ and $\Xi$.
The curve 
\[
\bar\Theta_t := (M_t)_\# \Theta, \qquad t \in [0,T],
\]
is a measure-valued spline, and it is optimal ---among all ways of moving $\Theta$ to $\Xi$--- in the sense of the total acceleration
\[
\begin{aligned}
&T\int_0^T \int_{\Gamma^2} |F_t|^2 d\Theta_t \, dt\, \\
&\partial_t \Theta_t + \bv \cdot \nabla_x \Theta_t + \nabla_v \cdot (F_t \Theta_t) = 0, \qquad  \Theta_0= \Theta, \quad  \Theta_T = \Xi. 
\end{aligned}
\]
The readers will recognise, in the last line, Vlasov's equations: the analogous of Newton's laws of motion $\bx'=\bv, \,\, \bv'=F$ in the case of probability measures.

\section{Our Desiderata}
\label{sec:desired_properties}
In this section, we describe the properties a divergence should possess.
For the sake of simplicity, we divide the properties into three subgroups: \begin{enumerate*}[label=(\roman*)]
    \item \textit{Analytical Properties}, \textit{i.e.} properties related to the regularity of the divergence,
    \item \textit{Topological Properties}, \textit{i.e.} properties related to the notion of convergence induced by the divergence, and
    \item \textit{Computational Properties}, \textit{i.e.} properties bounded to the computational aspectes of the divergence.
\end{enumerate*}
In what follows, $\mathcal{P}(\mathbb{R}^n)$ represents the set of probability measures supported on $\mathbb{R}^n$.
We say that a measure $\mu\in\PP(\erre^n)$ has finite $p$-th moment if 
\begin{equation}
    \int_{\erre^n}|\bx|^pd\mu(\bx)<+\infty.
\end{equation}
We denote with $\PP_p(\erre^n)$ the set of probability measures whose $p$-th moment is bounded.
Since every random vector in $\mathbb{R}^n$ is characterized only in terms of its associated probability distribution, we adopt a slight abuse of notation and use the random vector $\bX$ and its corresponding probability measure $\mu$ interchangeably.

\subsection{Analytical Properties}

First, a divergence should be symmetric and assign a null value only to two equal probability distributions.
These properties are ensured by requiring $\mathcal{D}$ to be a metric.

\begin{definition}[Metric]
    A function $\mathcal{D}:\PP(\erre^n)\times\PP(\erre^n)\to[0,\infty)$ is a metric if it satisfies the following three properties
    \begin{enumerate}[label=(\roman*)]
        \item \label{ax:dist1} $\mathcal{D}(\mu,\nu)=0$ if and only if $\mu=\nu$;
        \item \label{ax:dist2} $\mathcal{D}(\mu,\nu)=\mathcal{D}(\nu,\mu)$; and
        \item \label{ax:dist3} $\mathcal{D}(\mu,\nu)\le\mathcal{D}(\mu,\zeta)+\mathcal{D}(\zeta,\nu)$ for every $\zeta\in\PP(\erre^n)$.
    \end{enumerate}
\end{definition}

Identifying when two probability measures are equal is fundamental.
Equally important, however, is understanding how small changes in the arguments of a divergence affect its value.
To this end, the second property we investigate is the differentiability of a divergence measure with respect to one of its arguments.
Specifically, we focus on differentiability with respect to smooth families of parameterized probability distributions.

\begin{definition}[Differentiability]
    Let $I\subset\erre$ be an interval and $\{P_\theta\}_{\theta \in I}$ be a differentiable parametric family of probability distributions. We say that $\mathcal{D}$ is differentiable along the family $\{P_\theta\}_{\theta \in I}$ if the function
    \[
        \partial_\theta\mathcal{D}(P_\theta,Q):=\lim_{h\to0^+}\frac{\mathcal{D}(P_{\theta+h},Q)-\mathcal{D}(P_\theta,Q)}{h}
    \]
    is well-defined for every $\theta\in I$ and for every probability distribution $Q$.
\end{definition}

% 
% The last analytical property that we consider is convexity.
% % 
% Given a three random variables $X$, $Y$, and $Z$, we say that $\mathcal{D}$ is convex if 
% \begin{equation}
%     \mathcal{D}(X,\lambda Y+(1-\lambda)Z)\le \lambda\mathcal{D}(X, Y) + (1-\lambda) \mathcal{D}(X,Z)
% \end{equation}
% for every $\lambda\in[0,1]$ and every triplet $X,Y,$ and $Z$ in $\PP(\erre^n)$.
% 

\subsection{Topological Properties}

We then move to study the topological properties, \textit{i.e.} the notion of convergence that the divergence induces.
As we said previously, we want to determine when two distributions are similar.
In particular, given two divergences $\mathcal{D}$ and $\mathcal{D}'$, we are interested in whether we can bound one divergence with another.
This property is important as it allows to classify and group divergence functions based on the type of convergence they induce.

\begin{definition}[Equivalence]
    \label{def:equivalence}
    Given two divergences $ \mathcal{D}, \mathcal{D}':\PP(\erre^n)\times\PP(\erre^n)\to[0,\infty)$, we say that $ \mathcal{D}$ and $ \mathcal{D}'$ are equivalent if there exists two pairs of positive constants, namely $c_l,c_u\ge 0$ and $r,s$ such that
    \begin{equation}
        \label{eq:def_equivalence}
        c_l[ \mathcal{D}(\bX,\bY)]^r\le \mathcal{D}'(\bX,\bY)\le c_u[ \mathcal{D}(\bX,\bY)]^s
    \end{equation}
    for every couple of random vectors $\bX$ and $\bY$.
\end{definition}

\begin{remark}
    Notice that we have adopted a slightly more general notions of equivalence in which we are allowed to consider different powers of the divergences.
    This notion allows us to relate the convergence speed of two equivalent divergences.
    Indeed, if we have that if $\{\mu_n\}_{n\in\enne}$ is a sequence for which $\mathcal{D}(\mu_n,\mu)\le n^{-p}$, then, if $\mathcal{D}'$ satisfies \eqref{def:equivalence}, we have that $\mathcal{D}'(\mu_n,\mu)\le n^{-sp}$.
\end{remark}

% \textcolor{purple}{[GB: do we consider only polynomial rates between different metrics?]}

We then come to how the divergence behaves under linear operations.
First, we consider the sub-additivity by convolution, which regulates how the divergence behaves when we compare the sum of two random variables.

\begin{definition}[Sub-additivity by convolution]
\label{def:sub_add}
    Let $\bX$ and $\bY$ be two random vectors.
    A divergence $\mathcal{D}$ is sub-additive by convolution if 
    \begin{equation}
    \label{eq:sub_add}
        \mathcal{D}(\bX+\bZ,\bY+\bW) \le \mathcal{D}(\bX,\bY) + \mathcal{D}(\bZ,\bW),
    \end{equation}
    for any random vector $\bZ$ independent from $\bX$, and any random vector $\bW$ independent from $\bY$.
    Moreover, $D$ is sub-additive by  \textit{convex} convolutions if 
    \begin{equation}
    \label{eq:sub_add_conv}
        \mathcal{D}(\sqrt{\lambda}\bX+\sqrt{1-\lambda}\bZ,\sqrt{\lambda}\bY+\sqrt{1-\lambda}\bW) \le \sqrt{\lambda}\mathcal{D}(\bX,\bY) + \sqrt{1-\lambda}\mathcal{D}(\bZ,\bW),
    \end{equation}
    for every $\lambda\in(0,1)$, for any random vector $\bZ$ independent from $\bX$, and any random vector $\bW$ independent from $\bY$.
\end{definition}

Next, we consider the scale sensitivity of a divergence, which determines how the value of a divergence changes when we rescale the random vectors by a constant.

\begin{definition}[Scale Sensitivity]
    Given $p\in\enne$, we have that a divergence $\mathcal{D}$ is $p$-scale sensitive if
    \begin{equation}
        \mathcal{D}(c\bX,c\bY)=|c|^p\mathcal{D}(\bX,\bY)
    \end{equation}
    for every couple of random vectors $\bX$ and $\bY$ and every $c\in\erre$.
\end{definition}

\begin{definition}[Zoloratev Ideal divergence, Ref \cite{Zol}]
    A divergence is said to be Zoloratev ideal if it is both sub-additive by convolution and $1$-scale sensitive.
\end{definition}
% 

% \textcolor{purple}{[GB: emphasis + reference.]}
% 

% 
A particularly important case is the $0$-scale sensitivity, according to which scaling a variable does not change the value of the divergence.
Divergences enjoying such properties are necessary, for example, in economics, where data on the wealth of different populations may be collected in various currencies or measured in terms of different natural resources.
In such scenarios, it is essential to express data properties using indices that are independent of the unit of measurement.
Notable examples of such indices are the Gini Index and the Coefficient of Variation.\cite{ABGT,AGT,VN} 
Owing to the fact that a divergence acts on pairs of probability distributions, we consider a stronger notion of scale invariance, in which both arguments are scaled differently.

\begin{definition}[Scale Invariance]
    \label{def:scaleinvariance}
    Let $ \mathcal{D}:\PP(\erre^n)\times\PP(\erre^n)\to[0,+\infty)$ be a divergence function. 
    We say that $ \mathcal{D}$ is scale invariant if 
    \begin{equation}
        \label{eq:scaleinvariance}
         \mathcal{D}(Q\bX,Q'\bY)= \mathcal{D}(\bX,\bY)
    \end{equation}
    for every couple of random vectors $\bX$ and $\bY$ and any couple of diagonal 
    % \textcolor{purple}{[GB: only diagonal, right/ Note: introduce a macro for \emph{diag}.]} 
    matrices $Q,Q'$ whose diagonal values are positive, that is
    \[
        Q=  \diag(q_1,q_2,\dots,q_n):=\begin{pmatrix}
            q_1 & 0 & 0 & \dots & 0\\
            0 & q_2 & 0 & \dots & 0\\
            0 & 0 & q_3 & \dots & 0\\
            \vdots & \vdots & \vdots & \ddots & \vdots\\
            0 & 0 & 0 & \dots & q_n\\
        \end{pmatrix}
    \]
    and $Q'=\diag(q_1',q_2',\dots,q_n')$, with $q_i,q_i'>0$ for every $i=1,\dots,n$.
\end{definition}

We emphasize once again that the notion of scale invariance exhibited by both the Gini Index and the Coefficient of Variation differs from the one introduced in Definition~\ref{def:scaleinvariance}.
Indeed, while the Gini Index and the Coefficient of Variation are defined on a single probability measure, divergences are defined on pairs of probability measures.
For this reason, our notion of scale invariance accounts for all possible combinations of scalings that may be applied independently to each of the two arguments of the divergence.
% 
% 
% This distinction introduces additional complexity.
% % 
% Indeed, scale invariance is a stringent requirement demanding that the divergence function $\mathcal{D}$ remain unchanged under independent changes of units applied separately to each input distribution.
% 

\subsection{Computational Properties}

Finally, we examine the computational aspects of the divergences, \textit{i.e.}, those related to their implementation in applied problems.
In particular, we consider two aspects: 
\begin{enumerate*}
    \item the computational cost of evaluating the divergence between two discrete probability distributions and
    \item the how well the gradient of the divergence can be approximated when one of the two distributions is available only through samples.
\end{enumerate*}
% 
% In what follows, we assume that the probability measures at hand have a support containing $N$ points.
% 
% It is clear that the larger the number of points in the support $N$, the higher the number of operations we need to do compute the value of the divergence between two discrete measures supported over $N$ points.
% 

\begin{definition}[Complexity]
    We say that the divergence $\mathcal{D}$ is $O(N^p)$-complex if it takes at most $CN^p$ operations to compute the divergence between two probability distribution supported over $N$ points, where $C$ is a constant that does not depend on the probability measures at hand. 
\end{definition}

% 
% The complexity property quantifies the computational effort required to evaluate the divergence between probability distributions. 
% 

% 
% Another issue that is proper for applications is the approximation of the gradient when one of the two distributions to compare cannot be expressed in analytical terms.
% % 
% This happens in machine learning, where it is necessary to compare a generated probability distribution with an unknown target distribution accessible only via a set of samples, \textit{i.e.} the dataset. 
% % 
% In this case, we are bound to rely on the estimation of the gradient based on an empirical approximation of the distribution given by the available samples.
% % 
% Consequently, the gradient of the divergence function evaluated over a set of samples is expected to serve as an unbiased estimator of the true gradient.
% 

% 
A common practical challenge is approximating the gradient when one of the two distributions cannot be expressed analytically.
This situation often arises in machine learning, where a generated probability distribution must be compared to an unknown target distribution that is only accessible through samples—that is, the dataset.
In such cases, the gradient must be estimated from an empirical approximation of the distribution based on the available samples.
As a result, the gradient of the divergence computed on the samples is expected to provide an unbiased estimate of the true gradient.

\begin{definition}[Unbiased Gradient]
    \label{def:unbiasedgrad}
    Let $\{\nu_\theta\}_{\theta\in I}$ be a parametric family of distributions, where $I\subset\erre$ is an interval.
    We say that $\mathcal{D}:\mathcal{P}(\erre)\times\mathcal{P}(\erre)\to[0,\infty)$ has unbiased gradient if it satisfies the following property
    \begin{equation}
    \label{eq:def:unbiasedgrad}
        \mathbb{E}_{\bX^{(1)},\dots,\bX^{(N)}}\Big[\nabla_{\theta}\mathcal{D}(\mu_N,\nu_\theta)\Big]=\nabla_{\theta}\mathcal{D}(\mu,\nu_\theta),
    \end{equation}
    where $\mu$ is a probability distribution, $\mu_N$ is the probability distribution associated with $\frac{1}{N}\sum_{i=1}^N\delta_{\bX^{(i)}}$ where $\bX^{(1)},\dots,\bX^{(N)}$ are $N$ independent samples of $\mu$.
\end{definition}

% 
% The unbiased gradient property is important in every application that involves the optimization of an objective function, since these optimization tasks are carried out by first order methods.
% 

% % 
% As mentioned above, the optimal parameters of a neural network are obtained through optimization routines that require evaluating the divergence function multiple times.
% % 
% Therefore, the ability to compute the divergence between two distributions quickly is a critical property.
% % 
% Many divergence functions are computed by approximating an integral, such as the Entropy, the $l_2$ norm, the Energy Distance, and Fourier-based Metrics.
% % 
% In other cases, evaluating the divergence involves solving a minimization problem, as is the case with the Wasserstein Distance, which is notoriously time-consuming in most applications.
% % 

\section{Our Divergences of Interest}
\label{sec:metrics}

In this section, we introduce the Wasserstein Distances, the Fourier-based Metrics, and the Energy Distances.
We structure our discussion as follows.
We first introduce all three divergences and detail under which conditions they are well-defined.
Then, we study their equivalence and compare them in terms of the desired properties showcased in Section \ref{sec:desired_properties} in a dedicated section.

\subsection{The Wasserstein Distances}
 
% 
% We first consider the Wasserstein Distance.
% 
Albeit the Wasserstein Distance has been initially introduced in relation with transportation problems, it has become a prominent tool in several applied mathematical fields, ranging from economics to kinetic theory.
% 
% The first theoretical study about convergence to equilibrium of the Boltzmann equation is due to McKean Jr., who in Ref. \cite{McK} was able to find explicit rates of convergence towards the Maxwellian equilibrium for the Kac caricature of a Maxwell gas, a one-dimensional model introduced in the fifties of the last century by Mark Kac.\cite{Kac} 
% % 
% The pioneering paper by McKean contains a lot of enlightening remarks, and introduces the role of the entropy and of the Fisher information,\cite{KL} fruitfully used later on in different contexts.\cite{CJMTU}
% 

% 
The first usage of the Wasserstein Distance\footnote{Nowadays, the metric is named after the Russian mathematician Vasershtein (Wasserstein), who introduced it independently in a different field.\cite{Wass}} within kinetic theory was by Tanaka, who introduced it to study the equilibrium convergence for the Kac equation.\cite{Tan1}
More precisely, given two random vectors $\bX$ and $\bY$, Tanaka defined 
\[
    W^2_2(\mu,\nu) = \inf E[|X-Y|^2],
\]
% \textcolor{purple}{[GB: I'd write $\mathcal{W}(\mu,\nu)$ everywhre, if you agree.]}
where the infimum is taken over all pairwise couplings $(\bX,\bY)$ of one-dimensional random variables $X\sim \mu$ and $Y\sim\nu$.
Notice that $W^2_2(\mu,\nu)$ is well defined as long as the two probability distributions $\mu$ and $\nu$ have finite second moments.
Having in mind the application to kinetic theory, in Tanaka's work $\nu$ is a limiting Gaussian distribution with zero mean and variance equal to the variance of $\mu$ (also known as the \emph{best matching} Gaussian of $\nu$).
Hence, as for the relative entropy, Tanaka's functional  $W^2_2(\mu,\nu)$ measures how different a given distribution $\mu$ is from the stationary solution $\nu$, in our case a one-dimensional Gaussian distribution.

The properties of the Wasserstein Distance when applied to a pair of generic random vectors $(\bX,\bY)$
\begin{equation}
    \label{Wasser}
    W^2_2(\bX,\bY) = W^2_2(\mu,\nu) = \inf E[|\bX-\bY|^2]
\end{equation}
for any dimension $n\ge 1$, were then studied by Tanaka in an extension of his initial work\cite{Tan2} in which the author also investigated the connection between convergence to equilibrium for the Boltzmann equation with Maxwell molecules and the central limit problem in probability theory.\cite{RR}
% 
% In Ref. \cite{Tan2}, the author also investigates the connection between convergence to equilibrium for the Boltzmann equation with Maxwell molecules and the central limit problem in probability theory.\cite{RR}

More generally, a Wasserstein Distance of order $p$ was then defined as  follows.
\begin{definition}
    Given two probability measures $\mu,\nu \in \mathcal P_p(\R^n)$, the $p$-th Wasserstein Distance between $\mu$ and $\nu$ is defined as
    \begin{equation}
    \label{Was-p}
    W_p(\mu, \nu) = \inf_{\pi \in \Pi(\mu,\nu)}\left\{ \int_{\R^n\times\R^n} |\bx-\by|^p \, d\pi(\bx,\by)\right\}^{1/p}
    \end{equation}
    where $\pi$ runs over the set of transportation plans between $\mu$ and $\nu$, namely $\Pi(\mu,\nu)$, that is, the set of probability measures on $\R^n\times\R^n$ whose marginals are $\mu$ and $\nu$
    \[
    \int_{\R^n\times\R^n}\!\!\!\! \varphi(\bx) d\pi(\bx,\by) = \int_{\R^n} \varphi(\bx) d\mu(\bx)\quad\text{and}\quad \int_{\R^n\times\R^n}\!\!\!\! \varphi(\by) d\pi(\bx,\by) = \int_{\R^n} \varphi(\by) d\nu(\by)
    \]
    % and 
    % \[
    % \int_{\R^n\times\R^n} \varphi(\by) \, d\pi(\bx,\by) = \int_{\R^n} \varphi(\by) d\nu(\by),
    % \]
    for all $\varphi \in C_b(\R^n)$, the set of continuous and bounded functions on $\R^n$.
\end{definition}
Among the family of Wasserstein distances, we consider the $W_1$ distance, also known as the Kantorovich-Rubinstein distance, or the dual-Lipschitz norm. 
Indeed, owing to the Fenchel-Rockafellar duality principle,\cite{Vil03} $W_1$ can be rewritten as follows:
\begin{equation}
    \label{W1}
    W_1(\mu,\nu) = \sup \left\{\int_{\R^n} \varphi(\bx) d(\mu(\bx) - \nu(\bx)) \, : \varphi \in \mathrm{Lip}(\R^n), \|\varphi\|_{\mathrm{Lip}(\R^n)} \le 1 \right\}.   
\end{equation}
As we show in our study, the formulation \eqref{W1} is useful to study the relationships between the Wasserstein distance and other divergences.
For further details on the properties and applications of these distances, we refer to Refs. \cite{CaTo} and \cite{Vil03}.

\subsection{The Fourier-based Metrics}

The Fourier-based metrics are a family of divergences that acts on the space of phases of a probability distribution.
Given a probability measure $\mu$, we denote with $\hat\mu$ its Fourier transform, that is
\begin{equation}
    \label{eq:muFtransform}
    \hat\mu(\bxi)=\int_{\erre^n}e^{-i\bx\cdot\bxi}d\mu(\bx).
\end{equation}
The first Fourier-based Metric was introduced in \cite{GTW}, where the authors adopted the metric:
\[
\mathbb{F}_s(\mu,\nu)=\sup_{\bxi\in\erre^n}\frac{|\hat\mu(\bxi)-\hat\nu(\bxi))|}{|\bxi|^s}
\]
to study the tendency towards equilibrium for solutions of the spatially homogeneous Boltzmann equation for Maxwell molecules, allowing to prove that the solution to the Kac equation converges to its limit state exponentially.
% 
% convergence towards equilibrium for both Kac equation and the Boltzmann equation for Maxwell molecules was obtained in Ref. \cite{GTW}.
% 
 
% 
In a follow up work, these metrics were generalized, allowing to characterize $\mathbb{F}_s$ as the element corresponding to $p = \infty$ in a $L^p$-family of metrics.
We define the case $p=2$ as follows.

\begin{definition}
    For a given pair of random vectors $\bX$ and $\bY$ distributed according to $\mu$ and $\nu$, respectively, the $s$-Fourier-based Metric is defined as
    \begin{equation}
        \label{me-inf}
        \mathcal{F}_s(\bX,\bY) = \mathcal{F}_s(\mu,\nu) = \bigg[\int_{\erre^n} \frac{|\hat\mu(\bxi) - {\hat \nu}(\bxi)|^2}{|\bxi|^{s}}d\bxi\bigg]^{\frac{1}{2}},
    \end{equation}
    where $\hat\mu$ and $\hat\nu$ are the Fourier transform of $\mu$ and $\nu$, respectively.
\end{definition}

Notice that, differently from Wasserstein Distances, Fourier-based Metrics are not well-defined for any $s$ and any coupling of probability measures $\mu$ and $\nu$.
Indeed, given two probability distributions $\mu$ and $\nu$, let $l\in\enne$ be the smallest value for which 
\begin{equation}
    \label{eq:equal_l_moments}
     \int_{\erre^n}|\bx|^ld\mu \neq \int_{\erre^n}|\bx|^ld\nu.    
\end{equation}
We then have $\mathcal{F}_{s}^2(\mu,\nu)<+\infty$ if and only if the function 
\begin{equation}
    \label{eq:integral_function}
    \frac{|\hat\mu(\bxi)-\hat\nu(\bxi)|^2}{|\bxi|^s}    
\end{equation}
is integrable around $\bxi=0$.
If we express \eqref{eq:integral_function} in polar coordinates, we have that
\begin{align*}
    \int_{\erre^n}\frac{|\hat\mu(\bxi)-\hat\nu(\bxi)|^2}{|\bxi|^s}d\bxi&=\int_{S^{n-1}}\int_{0}^{+\infty}\frac{|\hat\mu(\rho,\theta)-\hat\nu(\rho,\theta)|^2}{\rho^s}\rho^{n-1}d\rho d\theta.
\end{align*}
Owing to \eqref{eq:equal_l_moments}, we have that $|\hat\mu(\rho,\theta)-\hat\nu(\rho,\theta)|^2\sim O(|\rho|^{2l})$.
Hence $\mathcal{F}^2_{s}(\mu,\nu)<+\infty$ if and only if
\[
    2l-s+n-1 >- 1 \quad\quad \Leftrightarrow \quad\quad l>\frac{s-n}{2}.
\]
therefore, in order to ensure that $\mathcal{F}_s(\mu,\nu)$ is finite, the probability distributions $\mu$ and $\nu$ must either
\begin{enumerate*}[label=(\roman*)]
    \item have equal moments up to $\Big[\frac{s-n}{2}\Big]$, where $\Big[\,.\,\Big]$ denotes the integer part function, if $\frac{s-n}{2} \not \in \N$; or
    \item have equal moments up to $\frac{s-n}{2}-1$ if $\frac{s-n}{2} \in \N$, where $n$ is the dimension of the space on which $\mu$ and $\nu$ are supported.  
\end{enumerate*}

% In the same paper, various relationships of this metric with other known metrics, including the Wasserstein Distance, allowed to obtain rates of convergence in the physical space. 
% % 
% The same Fourier-based Metric was subsequently used in Ref. \cite{TV} to prove uniqueness of the solution to the Boltzmann equation for Maxwell molecules without cut-off, as well as in Ref. \cite{CCG}, always in connection with the spatially homogeneous Boltzmann equation for Maxwell molecules and its representation in Wild sums.
% 

% 
% Given two probability measures $f,g \in \mathcal P(\R^n)$, and wo constants $p\ge 1$ and $s >0$ a Fourier-based Metric $D_{s,p}$ is defined by
% \begin{equation}\label{Fou-met}
% D_{s,p}(f,g) =\left\{ \int_{\R^n} \frac{|\hat f(\bxi) -\hat g(\bxi)|^p}{|\bxi|^{ps}}\, d\bxi\right\}^{1/p}.
% \end{equation}
% This class of metrics contains two interesting cases, which have been considered in the introduction. 

% The cases $p=2$ and $s =(n+\alpha)/2$ correspond to a multiple of the Energy Distance $\mathcal E_\alpha$ defined in \fer{f-energy}. As we shall see later, other values of $s$, corresponding to negative values of $\alpha$ can be considered.

% Next, the case $p \to +\infty$ coincides with the Fourier-based Metric $\mathbb{F}_s$ defined in \fer{me-inf}. 

\subsection{The Energy Distance}

% We now consider the Energy distance. 
% 
The notion of Energy Distance was introduced by Sz\'ekely around the mid-1980s, in several lectures he gave in Budapest (Hungary), in the Soviet Union, and --in  the US-- at the M.I.T., Yale, and Columbia University.\cite{Sze89,Sze03}
The idea behind the Energy Distance was to define a metric able to capture the structure of Newton's potential energy, that is:
\begin{equation}
    \label{eq:newtonspotential}
    N(\bx) = \int_{\R^3} \frac {1}{|\bx-\by|}d\mu(\by),
\end{equation}
which measures the energy necessary to move one unit mass from a location $\bx$ to infinity in a gravitational space with mass distribution $\mu$.
% 
% \begin{remark}
%     {\color{red}Note that the function \eqref{eq:newtonspotential}, can be naturally interpreted within the context of kinetic theory.
%     % 
%     Indeed, according to kinetic theory, interaction potentials range from weak to strong, according to the relationship
%     \[
%     \mathcal{K}(\nu, |\bv-\bw|) \approx |\bv-\bw|^\alpha, \quad -2 < \alpha \le 1,
%     \]
%     for weak potentials, and $\alpha>1$ for strong potentials.
%     % 
%     Notice that the Maxwellian pseudo-molecules are characterized by the exponent $\alpha=0$, while the rigid-spheres kernel is characterized by the exponent $\alpha =1$. 
%     % 
%     Hence the Energy Distance is related, according to the terminology of kinetic theory, to weak energy potentials.}
% \end{remark}
% 
% 
In their seminar work, Sz\'ekely and Rizzo considered the following family of distances
\begin{align}
    \label{eq:orig_endist}
    \mathcal{E}_\alpha(\mu,\nu) = &\;\;2\int_{\R^{2n}}|\bx-\by|^{\alpha} d\mu(\bx) d\nu(\by) - \int_{\R^{2n}}|\bx-\by|^{\alpha} d\mu(\bx) d\mu(\by)\\
        \nonumber&\;\; - \int_{\R^{2n}}|\bx-\by|^{\alpha} d\nu(\bx) d\nu(\by).
\end{align}
for any $\alpha\in(0,2)$ and proved that, for a suitable positive constant $C_\alpha$, it holds
\begin{equation}
    \label{eq:energy_fourier_generalintro}
    \mathcal{E}_\alpha(\mu,\nu) = C_\alpha \int_{\erre^n}\frac{|\hat\mu(\bxi)-\hat\nu(\bxi)|^2}{|\bxi|^{n+\alpha}}d\bxi = C_\alpha\mathcal{F}_{n+\alpha}^2(\mu,\nu).
\end{equation}
Ever since its introduction, the Energy Distance has been considered only for $\alpha\in(0,2)$.
Indeed, it is easy to see that if $\alpha=0$, $\mathcal{E}_0$ is always equal to $0$, whereas, when $\alpha=2$, the Energy Distance boils down to the Euclidean distance between the two mean vectors of $\mu,\nu$, which is not a metric.
Noticeably, the case $\alpha=-1$ has never been considered, although it is the rationale behind the initial definition.
In what follows, we show that, under suitable assumptions on the distributions $\mu$ and $\nu$, it is possible to extend $\mathcal{E}_\alpha$ to any $\alpha > -n$ with $\alpha\neq 2t$ for $t\in\enne$ while preserving the relation between $\mathcal{E}_\alpha$ and $\mathcal{F}_{n+\alpha}^2$ highlighted in \ref{eq:energy_fourier_generalintro}.
We start by introducing our extended definition of the Energy Distances.

\begin{definition}
\label{def:energy_dist_gen}
    Given a constant $\alpha\in\erre$, let $\bX\sim\mu$ and $\bY\sim\nu$ be two $n$-dimensional random variables, such that $\mu$ and $\nu$ have equal moments up to the $l$-th one, with $l\in\enne$, $l\ge \floor{\frac{\alpha}{2}}$.
    Then, we define the Energy Distance of order $\alpha$ between of $\mu$ and $\nu$ as follows.
    If $\alpha>0$, then
    \begin{equations}
    \label{en-dist}
        (-1)^{k}\Big(\mathcal{E}_\alpha(\mu,\nu)\Big)^2 = &\;\;2\int_{\R^{2n}}|\bx-\by|^{\alpha} d\mu(\bx) d\nu(\by) - \int_{\R^{2n}}|\bx-\by|^{\alpha} d\mu(\bx) d\mu(\by)\\
        &\;\; - \int_{\R^{2n}}|\bx-\by|^{\alpha} d\nu(\bx) d\nu(\by),
    \end{equations}
    where $k=\floor{\frac{\alpha}{2}}$.
    If $\alpha\in(-n,0)$, then
    \begin{equations}
    \label{en-dist-2}
        -\Big(\mathcal{E}_\alpha(\mu,\nu)\Big)^2 = &\;\;2\int_{\R^{2n}}|\bx-\by|^{\alpha} d\mu(\bx) d\nu(\by) - \int_{\R^{2n}}|\bx-\by|^{\alpha} d\mu(\bx) d\mu(\by)\\
        &\;\; - \int_{\R^{2n}}|\bx-\by|^{\alpha} d\nu(\bx) d\nu(\by).
    \end{equations}
\end{definition}
Before we proceed in our discussion, we need to stress three key points of our definition.
\begin{enumerate*}[label=(\roman*)]
\item First, in our case, the Energy Distance is defined as the square of a quantity. This choice is due to the fact that, in \eqref{eq:energy_fourier_generalintro} the right hand side of the equation is the square of a metric.
\item Second, we have that our extension has a sign term $(-1)^k$. This term naturally arises when we try to extend the identity  in \eqref{eq:energy_fourier_generalintro} to the case in which $\alpha>2$ or $-n<\alpha<0$.
\item Lastly, as $\alpha$ grows larger, we need $\mu$ and $\nu$ to have a larger amount of equal moments. Again, this restriction is due to the fact that the Fourier-based Metrics are not well-defined for any couple of probability measures.
\end{enumerate*}
% 

% 

% 
% Ever since its introduction, the Energy Distance has been considered only for $\alpha\in(0,2)$.
% % 
% Indeed, it is easy to see that if $\alpha=0$, $\mathcal{E}_0$ is always equal to $0$, whereas, when $\alpha=2$, the Energy Distance boils down to the Euclidean distance between the two mean vectors of $\mu,\nu$, which is not a metric.
% % 
% Noticeably, the case $\alpha=-1$ has never been considered, even though it is the rationale behind the initial definition in \fer{en-dist}.
% 

\begin{remark}
\label{rmrk:gini_energy_onedim}
    When $n=1$, \textit{i.e.} the random vectors are random variables, the Energy Distance of order $\alpha=1$ coincides with another well-known metric, the Cramér Distance.
    Given two random variables $X\sim\mu$ and $Y\sim\nu$, the Cramér Distance is defined as
    \begin{equation}
        \label{eq:Cramér_def}
        Cr(\mu,\nu)=\sqrt{\int_{\erre}(F_\mu(x)-F_\nu(x))^2dx},
    \end{equation}
    where $F_\mu$ and $F_\nu$ are the cumulative distribution function of $\mu$ and $\nu$, respectively.
    Indeed, the Cramér Distance can equivalently be written as
    \begin{align*}
        % \label{eq:Cramér_alternative}
        Cr(\mu,\nu)=2\mathbb{E}_{X\sim \mu, Y\sim \nu}[|X-Y|]-\mathbb{E}_{X,X'\sim \mu}[|X-X'|]-\mathbb{E}_{Y,Y'\sim \nu}[|Y-Y'|],
    \end{align*}
    where $X,X'$ are two independent random variables distributed according to $\mu$ and $Y,Y'$ are two independent random variables distributed according to $\nu$.
    Within this context, the Energy Distance is the natural way to extend the Cramér Distance to higher dimensional settings.\cite{Sze89,Sze03}
    Lastly, we notice that, on one dimensional probability distributions, the Energy Distance associated with $\alpha=1$, that is $\mathcal{E}_1$, is connected to the Gini Index.\cite{Gini1,Gini2}
    Indeed, given a random variable $\bX\sim\mu$ whose mean is non-null, the Gini Index associated with $\mu$ is defined as
    \begin{equation}
        \label{eq:def_gini}
        \mathcal{G}(\mu)=\frac{1}{2m_\mu}\int_{\erre\times\erre}|x-x'|d\mu(x)d\mu(x'),
    \end{equation}
    where $m_\mu$ is the non-null mean of $\mu$.
    Given $X\sim\mu$ and $Y\sim\nu$ two random variables with non null means, we then have that
    \begin{equation}
        \label{eq:energy:to_gini}
        \mathcal{E}_1(\mu,\nu)=2\GMD(\mu,\nu)-2m_\mu\mathcal{G}(\mu)-2m_\nu\mathcal{G}(\nu),
    \end{equation}
    where \begin{enumerate*}
        \item $m_\mu$ and $m_\nu$ are the means of $\mu$ and $\nu$, respectively; and
        \item $\GMD(\mu,\nu)=\int_{\erre\times\erre}|x-y|d\mu(x)d\nu(y)$ is the Gini Mean Difference. \footnote{$\GMD$ is also known as the Mean Absolute Difference.}
    \end{enumerate*}  
\end{remark}

\section{Comparing Divergences}

In this section, we establish several connections between the divergences introduced in Section \ref{sec:metrics}.  
In particular, we show that all these divergences are equivalent to each other.
Notice that some of the results we present are known in the literature.
However, we will either prove them through simpler arguments, or show how to extend the proof to a more general framework.
In reason of its importance in the field of machine learning, we will first consider the Energy Distance, which can be expressed both in the usual space, and in the Fourier space, as proven in Ref. \cite{SR}.

\subsection{The Energy Distance and the Fourier-based Metrics}
\label{sec:fourierandenergy}

In this section, we study the connections between the Energy Distances and the Fourier-based Metrics.
First of all, we notice that, if $\alpha\in(0,2)$, the definition in \eqref{en-dist} coincides with the original one given in \eqref{eq:orig_endist} up to a square factor.
We start our discussion on the Energy Distance by showing relation \eqref{eq:energy_fourier_generalintro} through a simpler argument.
We then show that the same argument is adaptable to handle the cases in which the power in \fer{en-dist} is negative and the case in which $\alpha>2$.

\begin{theorem}
\label{thm:positive_energy} 
Let $0<\alpha<2$. Then, the Energy Distance $\mathcal{E}_\alpha$ defined as in \fer{en-dist} is equal to the Fourier-based Metric \fer{me-inf} $\mathcal{F}_{n+\alpha}$, up to a multiplicative constant.
\end{theorem}

\begin{proof}
    Let $\bX$ and $\bY$ two $n$-dimensional random vectors, $n>1$ distributed according to $\nu$ and $\mu$, respectively.
    Given a constant $0<\alpha<2$, we set $ \lambda = 2-\alpha$, so that $0 < \lambda <2$.  
    Then, the Energy Distance of order $\alpha$, defined by \fer{en-dist}, between $\mu$ and $\nu$ can be rewritten as
    \begin{equations}
        \label{e-dist}
        \mathcal{E}^2_\alpha (\nu,\mu) &= 2\int_{\R^{2n}}|\bx-\by|^{2-\lambda}\, d\nu(\bx) d\mu(\by) +\\
        &\quad - \int_{\R^{2n}}|\bx-\by|^{2-\lambda}\, d\nu(\bx) d\nu(\by) - \int_{\R^{2n}}|\bx-\by|^{2-\lambda}\, d\mu(\bx) d\mu(\by)\\
        &=\int_{\R^{2n}}|\bx-\by|^{2-\lambda}\, d[\nu(\bx)-\mu(\bx)] d[\mu(\by)-\nu(\by)].
    \end{equations}
    
    If $\bX$ and $\bY$ have bounded second moments, we can equivalently write the Energy Distance as
    \begin{equations}
        \label{e-dis2}
        \mathcal{E}^2_\alpha(\nu,\mu) &= \sum_{k=1}^n \int_{\R^{2n}}W_{\lambda}(\bx-\by) \Big( x_k^2 d[\nu(\bx)-\mu(\bx)] d[\mu(\by)-\nu(\by)]\\
        &\quad+d[\nu(\bx)-\mu(\bx)] y_k^2 d[\mu(\by)-\nu(\by)]-2 x_k d[\nu(\bx)-\mu(\bx)] y_k d[\mu(\by)-\nu(\by)]\Big),
    \end{equations}
    where, for $\lambda>0$, we denoted
    \begin{equation}
        \label{wlambda}
        W_\lambda(\bx) = |\bx|^{-\lambda}.
    \end{equation}
    As shown in Ref. \cite{Stein}, Theorem 4.1, if $\lambda <n$ the Fourier transform of $W_\lambda$ equals
    \begin{equation}
        \label{fW}
        \widehat W_\lambda(\bxi) = \int_{\R^n} |\bx|^{-\lambda} e^{-i\bxi\cdot\bx} \, d\bx = \pi^{n/2}2^{n-\lambda} \frac{\Gamma\left(\frac{n-\lambda}2\right)}{\Gamma\left(\frac{\lambda}2\right)}\, \frac 1{|\bxi|^{n-\lambda}},
    \end{equation}
    where $\Gamma(\cdot)$ is the gamma function. Therefore, for $0<\lambda <n$, a simple computation gives
    \begin{equation}
        \label{Lapla}
        \Delta \widehat W_{\lambda}(\bxi) = (2-\lambda)\pi^{n/2}2^{n+1-\lambda} \frac{\Gamma\left(\frac{n-\lambda}2 +1\right)}{\Gamma\left(\frac{\lambda}2\right)}\, \frac 1{|\bxi|^{n-\lambda+2}}.
    \end{equation}
    Applying Parseval formula, we can easily find the Fourier transform of the right-hand side of \fer{e-dis2}.
    Indeed, we have that
    \[
        \int_{\R^n} x_k e^{-i\bxi\cdot\bx} \, d[\nu(\bx)-\mu(\bx)] = i \frac{\partial}{\partial \xi_k}[\hat\nu(\bxi)-\hat\mu(\bxi)],
    \]
    while
    \[
    \int_{\R^n} x_k^2 e^{-i\bxi\cdot\bx} \, d[\nu(\bx)-\mu(\bx)] = - \frac{\partial^2}{\partial \xi_k^2}[\hat\nu(\bxi)-\hat\mu(\bxi)].
    \]

    Likewise, we have that
    \[
    \int_{\R^n} y_k e^{i\bxi\cdot\by} \, d[\nu(\by)-\mu(\by)] = -i \frac{\partial}{\partial \xi_k}[\overline{\hat\nu(\bxi)-\hat\mu(\bxi)}]
    \]
    and
    \[
    \int_{\R^n} y_k^2 e^{i\bxi\cdot\by} \, d[\nu(\by)-\mu(\by)] = - \frac{\partial^2}{\partial \xi_k^2}[\overline{\hat\nu(\bxi)-\hat\mu(\bxi)}],
    \]
    where $\bar c$ denotes the conjugate of the complex number $c$. 
    To this point, we have simply to remark that in \fer{e-dis2} the term
    \[
    B_k(\by) = \int_{\R^n}W_{\lambda}(\bx-\by) x_k^2 d[\nu(\bx)-\mu(\bx)] \, d\bx 
    \]
    is the convolution of the two functions $W_{\lambda}(\bx)$ and $x_k^2[\nu(\bx)-\mu(\bx)]$, so that its Fourier transform is 
    \[
    \widehat B_k(\bxi) = -\widehat W_\lambda(\bxi)\frac{\partial^2}{\partial \xi_k^2}[\hat\nu(\bxi)-\hat\mu(\bxi)].
    \]
    Likewise, we handle the other terms in \fer{e-dis2}. 
    Therefore, by Parseval formula we get the expression in Fourier space of the Energy Distance, that reads
    \[
    \mathcal{E}^2_\alpha(\mu,\nu) = \frac 1{(2\pi)^n} \sum_{k=1}^n \int_{\R^n}\widehat W_\lambda(\bxi) \frac{\partial^2}{\partial \xi_k^2}\left[ (\hat\nu(\bxi)-\hat\mu(\bxi))(\overline{\hat\nu(\bxi)-\hat\mu(\bxi)})\right]\, d\bxi.
    \]
    Integrating twice by parts, we get
    \[
    \mathcal{E}^2_\alpha(\mu,\nu)= \frac 1{(2\pi)^n} \sum_{k=1}^n \int_{\R^n} \left|\hat\nu(\bxi)-\hat\mu(\bxi)\right|^2\frac{\partial^2}{\partial \xi_k^2} \widehat W_\lambda(\bxi)\, d\bxi,
    \]
    so that, taking into account \fer{Lapla}, and owing to the identity $\alpha=2-\lambda$, we conclude the thesis.
\end{proof}

\subsubsection{Negative Power Energy Distances}

The argument used in Theorem \ref{thm:positive_energy} to relate the Energy Distances with the Fourier-based Metrics can be extended to Energy Distances with negative indexes.
Indeed, let us consider $\mathcal{E}_{-\alpha}$ with $0 < \alpha < n$, then we have
% by suitably restricting the set of densities, we extend the class of Energy Distances to cases with negative indexes.
% 
% According to Definition \ref{def:energy_dist_gen}, let us consider $\mathcal{E}_{-\alpha}$ with $0 < \alpha < n$, then we have
% 
\begin{align}
    \nonumber\label{eq:negative_power}
    \mathcal{E}^2_{-\alpha}(\mu,\nu)&= \int_{\R^{2n}}|\bv-\bw|^{-\alpha}\, f(\bv) f(\bw)\, d\bv d\bw  + \int_{\R^{2n}}|\bv-\bw|^{-\alpha}\, g(\bv) g(\bw)\, d\bv d\bw \\
&\quad-2\int_{\R^{2n}}|\bv-\bw|^{-\alpha}\, f(\bv) g(\bw)\, d\bv d\bw,
\end{align}
where $f$ and $g$ are the densities of $\mu$ and $\nu$, respectively.
% 

% \begin{definition}
%     Given $0 < \alpha < n$, we define
%     \begin{align}
%         \nonumber\label{eq:negative_power}
%         \mathcal{E}^2_{-\alpha}(\mu,\nu)&= \int_{\R^{2n}}|\bv-\bw|^{-\alpha}\, f(\bv) f(\bw)\, d\bv d\bw  + \int_{\R^{2n}}|\bv-\bw|^{-\alpha}\, g(\bv) g(\bw)\, d\bv d\bw \\
%     &\quad-2\int_{\R^{2n}}|\bv-\bw|^{-\alpha}\, f(\bv) g(\bw)\, d\bv d\bw,
%     \end{align}
%     where $f$ and $g$ are the densities of $\mu$ and $\nu$, respectively.
% \end{definition}
% 

% 
Notice that when we consider the Energy Distance induced by a negative index, we need to swap the sign in the definition in \eqref{en-dist}. 
Moreover, owing to the Hardy-Littlewood-Sobolev inequality,\cite{Lieb83} the negative Energy Distance is well-defined if both $\mu$ and $\nu$ belong to $\PP_p(\erre^n)$, with $p = 2n/(2n-\alpha)$.

% 
% $L^p(\R^n)$, with $p = 2n/(2n-\alpha)$.
% 

\begin{theorem}
    % Let $\mu$ and $\nu$ belong to $L^p(\R^n)$, 
    % 
    Let $\mu$ and $\nu$ belong to $\PP_p(\erre^n)$,with $p = 2n/(2n-\alpha)$. Then the negative Energy Distance $\mathcal{E}_{-\alpha}$ is a well defined divergence for every $0<\alpha<n$.
    Moreover, we have that 
    \begin{equation}
        \label{eq:fourier_negative_energy}
        \mathcal{E}^2_{-\alpha}(\mu,\nu)=\pi^{n/2}2^{n-\alpha} \frac{\Gamma\left(\frac{n-\alpha}2\right)}{\Gamma\left(\frac{\alpha}2\right)}\int_{\erre^n}\frac{|\hat f(\bxi)- \hat g(\bxi)|^2}{|\bxi|^{n-\alpha}} d\bxi.
    \end{equation}
\end{theorem}

\begin{proof}
    For the sake of exposition, assume that $\mu$ and $\nu$ are induced by a density function, namely $f$ and $g$, respectively.
    Following the same argument used in the proof of Theorem \ref{thm:positive_energy}, we have
    \[
        \mathcal{E}^2_{-\alpha}(\mu,\nu)=\int_{\erre^n}\int_{\erre^n}|\bx-\by|^{-\alpha}(f(\bx)-g(\bx))(f(\by)-g(\by))d\bx d\by.
    \]
    If we set
    \[
    B(\by)=\int_{\erre^n}|\bx-\by|^{-\alpha}(f(\bx)-g(\bx))d\bx
    \]
    we infer, by Parseval's formula that
    \begin{align}
        \mathcal{E}^2_{-\alpha}(\mu,\nu)=\frac 1{(2\pi)^n}\int_{\erre^n}\hat B(\bxi)(\overline{\hat f(\bxi)- \hat g(\bxi)}) d\bxi
    \end{align}
    where $\hat B$, $\hat f$, and $\hat g$ are the Fourier transform of $B$, $f$, and $g$, respectively.
    Since by definition $B(\by)$ is a convolution between $|\bx|^{-\alpha}$ and $f(\bx)-g(\bx)$, using identity \eqref{fW} we infer that \fer{eq:fourier_negative_energy} holds true.

    Lastly, we notice that the metric properties of $\mathcal{E}_{-\alpha}$ follow from identity \eqref{eq:fourier_negative_energy}, since the right hand side is a Fourier-based Metric.
\end{proof}

\subsubsection{Higher Order Energy Distances}
We now consider the Energy Distances defined by indexes $\alpha$ that are larger than $2$.
As for negative indexes, we leverage the results from Lieb to study this extension and relate it to the Fourier-based Metrics.
Following Definition \ref{def:energy_dist_gen}, let us consider the Energy Distance
\begin{equation}
    \label{eq:energyalpahbig}
    \mathcal{E}^2_\alpha(\mu,\nu)=(-1)^{k}\int_{\erre^n\times\erre^n}|\bx-\by|^\alpha d[\mu(\bx)-\nu(\bx)]d[\nu(\by)-\mu(\by)],
\end{equation}
where $k=\floor{\frac{\alpha}{2}}$.
For any $\alpha > 2$, there exists a $0<\lambda<2$ for which $\alpha+\lambda=2(k+1)$.
Therefore, we write
\begin{equation}
    \label{eq:high_order_rearrangement}
    (-1)^{k}\mathcal{E}^2_\alpha(\mu,\nu)=\int_{\erre^n\times\erre^n}\frac{1}{|\bx-\by|^\lambda}\Big(\sum_{i=1}^n(x_i-y_i)^2\Big)^{k+1}d[\mu(\bx)-\nu(\bx)]d[\nu(\by)-\mu(\by)].
\end{equation}
Notice that, if $k=0$, we go back to the case studied in Section \ref{sec:fourierandenergy}.
Likewise, we have that by expanding the term $(\sum_{i=1}^n(x_i-y_i)^2)^{k+1}$, we can rewrite the right hand side of \eqref{eq:high_order_rearrangement} as
\[
\sum_{j\in J}\int_{\erre^n\times\erre^n}\frac{1}{|\bx-\by|^\lambda}d[p_j(\bx)(\mu(\bx)-\nu(\bx))]d[q_j(\by)(\nu(\by)-\mu(\by))],
\]
where $p_j$ and $q_j$ are polynomials with respect to the variables $\bx$ and $\by$, respectively, and $J$ is a suitable set of indexes.
By iterating the argument used for $k=1$, we get
\begin{equations}
    &\Big[{p_j(\bx)(\mu(\bx)-\nu(\bx))}\Big]=(-1)^{d_p}\partial_{J_p} (\hat\mu(\bxi)-\hat\nu(\bxi)),\\
    &\Big[\overline{\widehat{q_j(\by)(\nu(\by)-\mu(\by)}})\Big]=(-1)^{d_q}\partial_{J_q}\overline{(\hat\mu(\bxi)-\hat\nu(\bxi))},
\end{equations}
where $d_{p_j}$ and $d_{q_j}$ are the degree of the polynomials $p_j$ and $q_j$, respectively, $\partial_{J_{p_j}}$ and $\partial_{J_{q_j}}$ are the derivative with respect to a suitable combination of $\xi_i$ entries that depends on the polynomials $p_j$ and $q_j$, respectively.
Owing to the properties of derivatives, we have that
\begin{align*}
    \int_{\R^n\times\R^n}&W_\lambda(\bx-\by)\Big(\sum_{i=1}^n(x_i-y_i)^2\Big)^{k+1}d[\mu(\bx)-\nu(\bx)]d[\nu(\by)-\mu(\by)]\\
    &=-\int_{\R^n\times\R^n} W_\lambda(\bx-\by)\sum_{(i_1,\dots,i_{k+1})} \Big(\prod_{l=1}^{k+1}(x_{i_l}-y_{i_l})^2\Big)d[\mu(\bx)-\nu(\bx)]d[\mu(\by)-\nu(\by)]\\
    &=-\sum_{(i_1,\dots,i_{k+1})}\int_{\R^n}\hat W_\lambda(\bxi) \Big((-1)^{k+1}\partial^2_{\xi_{i_1}}\dots\partial_{\xi_{i_{k+1}}}^2 |\hat\mu(\bxi)-\hat\nu(\bxi)|^2\Big)d\bxi\\
    &=(-1)^{k}\int_{\R^n}\hat W_\lambda(\bxi)\Big(\Delta^{k+1} |\hat\mu(\bxi)-\hat\nu(\bxi)|^2\Big)d\bxi\\
    &=(-1)^{k}\int_{\R^n}(\Delta^{k+1}\hat W_\lambda(\bxi))|\hat\mu(\bxi)-\hat\nu(\bxi)|^2d\bxi\\
    &=(-1)^{k}\int_{\R^n}\frac{C_\alpha}{|\bxi|^{n+2(k+1)-\lambda}}|\hat\mu(\bxi)-\hat\nu(\bxi)|^2d\bxi,
\end{align*}
where $C_\alpha\ge 0$ is a suitable constant.
In particular, we have that
\begin{equation}
    \mathcal{E}^2_\alpha(\mu,\nu)=C_\alpha\mathcal{F}^2_{n+\alpha}(\mu,\nu).
\end{equation}
Notice that, in order to have that $\mathcal{E}_\alpha(\mu,\nu)\ge 0$, we have that
\begin{equation}
    \mathcal{E}_\alpha(\mu,\nu)=\begin{cases}
        \int_{\erre^n\times\erre^n}|\bx-\by|^\alpha_2d[\mu(\bx)-\nu(\bx)]d[\nu(\by)-\mu(\by)]\quad\quad&\text{if}\;\alpha\in A\\
        \int_{\erre^n\times\erre^n}|\bx-\by|^\alpha_2d[\mu(\bx)-\nu(\bx)]d[\mu(\by)-\nu(\by)]\quad\quad&\text{otherwise},
    \end{cases}
\end{equation}
where $A=\bigcup_{k=1}(4k-4,4k-2)$.
In particular, to make $\mathcal{E}_\alpha$ always positive, the sign of the right hand side of \eqref{eq:energyalpahbig} changes depending on the value of $\alpha$.

We summarise our findings in the following theorem.

\begin{theorem}
    The Energy Distance $\mathcal{E}_\alpha$ for any $\alpha>0$, defined as
    \begin{equation}
        \mathcal{E}_\alpha(\mu,\nu)=(-1)^{\floor{\frac{\alpha}{2}}+1}\int_{\erre^n\times\erre^n}|\bx-\by|_2^\alpha d[\mu(\bx)-\nu(\bx)]d[\mu(\by)-\nu(\by)],
    \end{equation}
    where $\floor{x}=\max\{k\in\N,\;s.t.\;k\le x\}$ is well-defined for any couple of probability measures that have the same moments up to the $\floor{\frac{3k-\alpha}{2}}$-th order, where $k=\floor{\frac{\alpha}{2}}$.
\end{theorem}

\begin{remark}
    Notice that if $\alpha=2k$ with $k\in\N$, the Energy Distance $\mathcal{E}_{2k}$ is no longer a metric as it only depends on the moments of the marginals of the two measures.
    Indeed, given $\mu$ and $\nu$ two probability measures, we have that
    \begin{align*}
        \mathcal{E}_{2k}(\mu,\nu)&=\int_{\erre^n\times\erre^n}\Big(\sum_{i=1}^n(x_i-y_i)^2\Big)^kd[\mu(\bx)-\nu(\bx)]d[\mu(\by)-\nu(\by)]\\
        &=\sum_{j\in J}\bigg(\int_{\erre^n}p_j(\bx)d[\mu(\bx)-\nu(\bx)]\bigg)\bigg(\int_{\erre^n}q_j(\by)d[\mu(\by)-\nu(\by)]\bigg),
    \end{align*}
    where $p_j$ and $q_j$ are suitable polynomials whose degree is at most equal to $2k$ and $J$ is a suitable set of indexes.
    In particular, we have that $\mathcal{E}_{2k}(\mu,\nu)$ depends only on the first $2k$ moments of the measures $\mu$ and $\nu$.
    Since $\mathcal{E}_{2k}(\mu, \mu)=0$, it follows that $\mathcal{E}_{2k}(\mu, \nu)=0$ for any pair of probability measures $\mu$ and $\nu$ that share the same moments up to order $2k$, hence $\mathcal{E}_{2k}$ is not a metric.
\end{remark}

Lastly, we study the sub-additivity by convolution property.
First, we show that every $\mathcal{E}_\alpha$ is sub-additive by convolution.

\begin{theorem}
    For every $\alpha>-n$, the divergence $\mathcal{E}_\alpha$ is sub-additive by convolution.
\end{theorem}

\begin{proof}
    Let $\bX_1\sim\mu_1$ and $\bY_1\sim\nu_1$ be two random vectors.
    Moreover, let $\bX_2\sim\mu_2$ and $\bY_2\sim\nu_2$ be two random vectors such that $\bX_1$ is independent from $\bX_2$ and $\bY_1$ is independent from $\bY_2$.
    Owing to the fact that $\bX_1+\bX_2\sim \mu_1\ast\mu_2$ and $\bY_1+\bY_2\sim \nu_1\ast\nu_2$, we have that
    \begin{align*}
        \mathcal{E}_\alpha(\bX_1+\bX_2,&\bY_1+\bY_2)=C_\alpha\mathcal{F}_{n+\alpha}(\bX_1+\bX_2,\bY_1+\bY_2)\\
        &=C_\alpha\sqrt{\int_{\erre^n}\frac{|\hat\mu_1(\bxi)\hat\mu_2(\bxi)-\hat\nu_1(\bxi)\hat\nu_2(\bxi)|^2}{|\bxi|^{n+\alpha}}d\bxi}\\
        &=C_\alpha\sqrt{\int_{\erre^n}\frac{|\hat\mu_1(\bxi)\hat\mu_2(\bxi)-\hat\nu_1(\bxi)\hat\mu_2(\bxi)+\hat\nu_1(\bxi)\hat\mu_2(\bxi)-\hat\nu_1(\bxi)\hat\nu_2(\bxi)|^2}{|\bxi|^{n+\alpha}}d\bxi}\\
        &\le C_\alpha\sqrt{\int_{\erre^n}\frac{|\hat\mu_2(\bxi)|^2|\hat\mu_1(\bxi)-\hat\nu_1(\bxi)|^2}{|\bxi|^{n+\alpha}}d\bxi} +\sqrt{\int_{\erre^n}\frac{|\hat\nu_1(\bxi)|^2|\hat\mu_2(\bxi)-\hat\nu_2(\bxi)|^2}{|\bxi|^{n+\alpha}}d\bxi}\\
        &\le C_\alpha\Big(\mathcal{F}_{n+\alpha}(\bX_1,\bY_1)+\mathcal{F}_{n+\alpha}(\bX_2,\bY_2)\Big)\\
        &=\mathcal{E}_\alpha(\bX_1,\bY_1) + \mathcal{E}_\alpha(\bX_2,\bY_2),
    \end{align*}
    which concludes the thesis.
\end{proof}

Lastly, we notice that, by taking $\alpha$ large enough, it is possible to retrieve an Energy Distance that is sub-additive by convex convolution.

\begin{theorem}
    When $\alpha\ge 4$, $\mathcal{E}_\alpha$ is sub-additive by convex convolution.
\end{theorem}

\begin{proof}
    Given $\lambda\in(0,1)$ and four random vectors $\bX_1,\bX_2,\bY_1$, and $\bY_2$ as in Definition \ref{def:sub_add}, we have that
    \begin{align*}
        &\mathcal{E}_\alpha(\sqrt{\lambda}\bX_1+\sqrt{1-\lambda}\bX_2,\sqrt{\lambda}\bY_1+\sqrt{1-\lambda}\bY_2)\\
        &\quad\quad\quad\quad\quad\quad\quad\quad\quad\le\mathcal{E}_\alpha(\sqrt{\lambda}\bX_1,\sqrt{\lambda}\bY_1)+\mathcal{E}_\alpha(\sqrt{1-\lambda}\bX_2,\sqrt{1-\lambda}\bY_2)\\
        &\quad\quad\quad\quad\quad\quad\quad\quad\quad=\lambda^{\frac{\alpha}{4}}\mathcal{E}_\alpha(\bX_1,\bY_1)+(1-\lambda)^{\frac{\alpha}{4}}\mathcal{E}_\alpha(\bX_2,\bY_2),
    \end{align*}
    To conclude the thesis, we notice that $\lambda\ge\lambda^{\frac{\alpha}{4}}$ and $1-\lambda\ge(1-\lambda)^{\frac{\alpha}{4}}$ if and only if $\alpha\ge 4$, which concludes the thesis.
\end{proof}

\begin{remark}
    Notice that to ensure that $\mathcal{E}_\alpha(\mu,\nu)<+\infty$ when $\alpha \ge 4$, we need to ask $\mu$ and $\nu$ to have equal mean and covariance matrix.
\end{remark}

\subsubsection{Comparing the Fourier-based Metrics and the Energy Distances}

Owing to Theorem \ref{thm:positive_energy}, we have that for $-n<\alpha$ and $\alpha\neq2k$ for every $k\in\enne$, the Energy Distance is, up to a multiplicative factor, equal to a suitable Fourier-based Metric.
In particular, for every couple of probability measures $\mu$ and $\nu$ it holds
\begin{equation}
\label{eq:F=ED}
   \mathcal{E}_{\alpha}(\mu,\nu) = C_\alpha \mathcal{F}_s(\mu,\nu),
\end{equation}
where $\alpha$ and $s$ are such that $s=n+\alpha$ and $n$ is the dimension on which the random vectors are supported.
% related by the following relationship
% \begin{equation}
%     % \label{eq:rel_alpha_to_s}
%     s=n+\alpha,
% \end{equation}
% where $n$ is the dimension on which the random vectors are supported.
% 
Owing to \eqref{eq:F=ED}, we are then able to give a full profiling of the Fourier-based Metrics and the Energy Distances in terms of the properties highlighted in Section \ref{sec:desired_properties}.
% % 

\textbf{Analytical Properties.} Owing to the identity \eqref{eq:F=ED}, we trivially infer that the Energy Distance is a metric whenever $\alpha>-n$ and $\alpha \neq 2k$ for $k\in\enne$, since $\mathcal{F}_s$ is a metric for every $s>0$.
Similarly, the Fourier-based Metrics are differentiable since so is the Energy Distance.
\textbf{Topological Properties.} Both $\mathcal{E}_\alpha$ and $\mathcal{F}_{n+\alpha}$ are $\frac{\alpha}{2}$ scale sensible thus none of the two divergences is scale invariant.
Lastly, we have proven that if $\alpha>4$ $\mathcal{E}_\alpha$ is sub-additive by convex convolution, implying that also $\mathcal{F}_s$ is sub additive by convex convolution when $s\ge n+4$ in accordance to previous results.
{\textbf{Computational Properties.}} Any Fourier-based Metric for which $s=n+\alpha$ for $\alpha>-n$ has unbiased gradient since $\mathcal{E}_{\alpha}$ possess this property for every $\alpha>-n$.
Indeed, let $\{P_\theta\}_{\theta\in I}$ be a family of differentiable probability distributions and let $\{\bX^{(i)}\}_{i=1,\dots,N}$ be $N$ i.i.d. random vectors distributed as $\bX\sim\mu$.
Denoted with $\mu_N$ the empirical distribution associated with the samples $\{\bX^{(i)}\}_{i=1,\dots,N}$, we have
\[
\partial_\theta \mathcal{E}^2_\alpha(\mu_N,P_\theta)=\partial_\theta\frac{1}{N}\sum_i^N\int_{\erre^n}|\bX^{(i)}-\by|P_\theta(d\by)+\partial_\theta\int_{\erre^n}\int_{\erre^n}|\by-\by'|P_\theta(d\by)P_\theta(d\by')
\]
since the term $\frac{1}{N^2}\sum_{i,j}^N|\bX^{(i)}-\bX^{(j)}|$ does not depend on $\theta\in I$.
By taking the expected value of $\partial_\theta \mathcal{E}^2_\alpha(\mu_N,P_\theta)$ we obtain
\[
\mathbb{E}\Big[\partial_\theta \mathcal{E}^2_\alpha(\mu_N,P_\theta)\Big]=\int_{\erre^n}\int_{\erre^n}|\bx-\by|\partial_\theta P_\theta(d\by)\mu(d\bx)+\partial_\theta\int_{\erre^n}\int_{\erre^n}|\by-\by'|P_\theta(d\by)P_\theta(d\by'),
\]
which is equal to $\partial_\theta \mathcal{E}_\alpha^2(\mu,P_\theta)$, thus proving that $\mathcal{E}_\alpha$ has unbiased gradient.
To conclude, we notice that the Energy Distance appears preferable, as its complexity is lower.
Indeed, computing the Energy Distance requires only the sum of $O(N^2)$ terms and avoids the need to compute two Fourier transforms.

\subsection{The Wasserstein Distance and the Fourier-based Metric}

In this section, we study the relation between the Wasserstein Distance and the Energy Distance.
Owing to Theorem \ref{thm:positive_energy}, studying the relation between the Wasserstein Distance and the Energy Distance directly translate to a study of the relations between the Wasserstein Distances and the Fourier-based Metrics.
In particular, we leverage Theorem \ref{thm:positive_energy} to propose a novel and simpler proof of the equivalence between the Fourier and the Wasserstein Distances over the set of probability measures with uniformly bounded second momentum.
% \textcolor{purple}{[GB: here equivalence is misleading. I would call it \emph{local/conditional equivalence}. They do not induce the same topology on the whole space.]}

\begin{theorem}
    Given two absolutely continuous probability distributions, namely $\mu$ and $\nu$, we have that
    \begin{equation}
        \frac{1}{2}\mathcal{E}^2_1(\mu,\nu)\le W_1(\mu,\nu)
    \end{equation}
    and
    \begin{equation}
        \frac{\Gamma\left(\frac{n+1}2\right)}{2\pi^{(n+1)/2}}\mathcal{F}^2_{n+1}(\mu,\nu)\le W_1(\mu,\nu).
    \end{equation}
\end{theorem}

\begin{proof}
    Since $\mu$ and $\nu$ are absolutely continuous distributions, there exists two functions, namely $T$ and $S$ such that the following conditions hold
    \begin{itemize}
        \item $T_\#\mu=\nu$ and $S_\#\nu=\mu$;
        \item $T\circ S= S \circ T= \Id$, where $\Id$ is the identity function on $\erre^n$; and
        \item $W_1(\mu,\nu)=\int_{\erre^{n}}|\bx-T(\bx)|d\mu=\int_{\erre^{n}}|\by-S(\by)|d\nu$.
    \end{itemize}
    In particular, $T$ is the optimal transportation map from $\mu$ to $\nu$ and $S$ is the optimal transportation map from $\nu$ to $\mu$.
    We then have the following 
    \begin{align}
    \label{eq:ineqT}
        \nonumber\int_{\erre^n\times\erre^n}&|\bx-\by|\mu(d\bx)\nu(d\by)=\int_{\erre^n\times\erre^n}|\bx-T(\bx)+T(\bx)-\by|\mu(d\bx)\nu(d\by)\\
        \nonumber&\le \int_{\erre^n\times\erre^n}\big(|\bx-T(\bx)|+|T(\bx)-\by|\big)\mu(d\bx)\nu(d\by)\\
        \nonumber&=\int_{\erre^n}\int_{\erre^n}|\bx-T(\bx)|\mu(d\bx)\nu(d\by)+\int_{\erre^n}\int_{\erre^n}|T(\bx)-\by|\mu(d\bx)\nu(d\by)\\
        &=\int_{\erre^n}|\bx-T(\bx)|\mu(d\bx)+\int_{\erre^n}\int_{\erre^n}|\by'-\by|\nu(d\by')\nu(d\by)
    \end{align}
    since, owing to the fact that $T_\#\mu=\nu$, we have that 
    \begin{equation}
        % \label{eq:ineqT}
        \int_{\erre^n}\int_{\erre^n}|T(\bx)-\by|\mu(d\bx)\nu(d\by)=\int_{\erre^n}\int_{\erre^n}|\by'-\by|\nu(d\by')\nu(d\by).
    \end{equation}
    Similarly, by using $S$ instead of $T$, we have that
    \begin{equation}
        \label{eq:ineqS}
        \int_{\erre^n\times\erre^n}|\bx-\by|\mu(d\bx)\nu(d\by)\le \int_{\erre^n}|\by-S(\by)|\nu(d\by)+\int_{\erre^n}\int_{\erre^n}|\bx'-\bx|\mu(d\bx')\mu(d\bx),    
    \end{equation}
    Since $W_1(\mu,\nu)=\int_{\erre^n}|\by-S(\by)|\nu(d\by)$, summing up the inequalities \eqref{eq:ineqT} and \eqref{eq:ineqS}, we have
    \begin{align}
        \label{eq:ineq_tot}
        \nonumber2\int_{\erre^n\times\erre^n}|\bx-\by|\mu(d\bx)\nu(d\by)&\le 2W_1(\mu,\nu)+\int_{\erre^n}\int_{\erre^n}|\bx'-\bx|\mu(d\bx')\mu(d\bx)\\
        &\quad+ \int_{\erre^n}\int_{\erre^n}|\by'-\by|\nu(d\by')\nu(d\by),
    \end{align}
    which, by rearranging the terms and recalling the definition of the Energy Distance, we get
    \begin{align}
        \mathcal{E}_1^2(\mu,\nu)\le 2W_1(\mu,\nu),
    \end{align}
    hence
    \begin{equation}
        \frac{\Gamma\left(\frac{n+1}2\right)}{2\pi^{(n+1)/2}}\mathcal{F}^2_{n+1}(\mu,\nu)\le W_1(\mu,\nu)
    \end{equation}
    since $\Gamma(\frac{1}{2})=\sqrt{\pi}$.
\end{proof}

Through a different argument based on the dual representation of the Wasserstein Distance, it is possible to prove that the Wasserstein Distance is always lower than a suitable power root of a Fourier-based Metric.

\begin{theorem}
    For all $M>0$, there exists a constant $C_{M,n}$ that depends only on $n$ and $M$ such that, for all pairs of probability distributions $\mu,\nu\in\PP(\erre^n)$ with bounded second moment satisfying
    \[
        \int_{\erre^n}|x|^2d(\mu-\nu)<M,
    \]
    it holds true that
    \begin{equation}
        W_1(\mu,\nu)\le C_{M,n} \left( \mathcal{F}_{n+1}(\mu,\nu)^{\frac{2}{n+2}} + \mathcal{F}_{n+1}(\mu,\nu)^{\frac{4}{(n+2)^2} }\right).
    \end{equation}
\end{theorem}

\begin{proof}
For the sake of simplicity, we assume that $\mu$ and $\nu$ are induced by two density functions, namely $f$ and $g$.
Invoking Proposition 2.15 from a book by J.-A.~Carrillo and G.~Toscani\cite{CaTo}, we have, for all $m \geq 1$, and for a constant $C=C(m,n)$, that
\begin{equation}
\label{eq:CaTo}
W_1(\mu,\nu) \leq C (\|f-g\|^\star_m + (\|f-g\|^\star_m)^{1/m}),
\end{equation}
where 
\begin{equation}\label{starm}
\|f-g\|^\star_m = \sup_{\phi \in \mathrm{C}_b^m(\erre^n)} \int_{\erre^n} \phi \, d(\mu-\nu), \quad \text{with} \,\, \|\phi\|_{\mathrm{C}^m(\erre^n)} \leq 1.
\end{equation}
Let $\phi$ be admissible in \eqref{starm}. Without loss of generality, let $\phi(0)=0.$ 
Let us fix $R>0$, and take a smooth cut-off function $\chi_R$ such that 
\begin{equation}
    \chi_R(x)=\begin{cases}
        1\quad\quad\text{if}\;\;|\bx|\le R\\
        0\quad\quad\text{if}\;\;|\bx|\ge R+1\\
        \in(0,1)\quad\quad\text{otherwise}
    \end{cases}.
\end{equation}
and $|\partial^\alpha_{x_j}\chi_R|\le C$ for every $j=1,\dots,n$, and $1\leq \alpha \leq m,$ and $C=C(m,n)$.
We then write
\begin{equation}
    \label{eq:potential_split}
    W_1(\mu,\nu)=\int_{\erre^n}\phi(\bx)(1-\chi_R(\bx))(f(\bx)-g(\bx))d\bx+\int_{\erre^n}\phi(\bx)\chi_R(\bx)(f(\bx)-g(\bx))d\bx.
\end{equation}

First, we notice that
\begin{align}
\nonumber\int_{\erre^n}\phi(\bx)(1-\chi_R(\bx))(f(\bx)-g(\bx))d\bx&\le \int_{|\bx|\ge R}|\phi(\bx)|(f(\bx)+g(\bx))d\bx\\
    &\le \frac{C}{R}\int_{|\bx|\ge R}|\bx|^2(f(\bx)+g(\bx))d\bx\\
    \nonumber&\le \frac{2CM}{R}
\end{align}
by assumption and by the properties of $\phi_R$.
We now move to the compactly supported term in \eqref{eq:potential_split}.
In this case, we have
\begin{align*}
    \int_{\erre^n}\phi(\bx)\chi_R(\bx)(f(\bx)-&g(\bx))d\bx=\int_{\erre^n}\widehat{\phi\chi_R}(\bxi)(\overline{\hat{f}(\bxi)-\hat{g}(\bxi)})d\bxi\\
    &=\int_{\erre^n}|\bxi|^{\frac{n+1}{2}}\widehat{\phi\chi_R}(\bxi)\frac{\overline{\hat{f}(\bxi)-\hat{g}(\bxi)}}{|\bxi|^{\frac{n+1}{2}}}d\bxi\\
    &\le\bigg(\int_{\erre^n}|\bxi|^{n+1}|\widehat{\phi\chi_R}(\bxi)|^2d\bxi\bigg)^\frac{1}{2}\bigg(\int_{\erre^n}\frac{|\hat{f}(\bxi)-\hat{g}(\bxi)|^{2}}{|\bxi|^{n+1}}d\bxi\bigg)^\frac{1}{2}\\
    %&=\bigg(\sum_{j=1}^n\int_{\erre^n}|i\xi_j\wideha{u\phi_R}(\bxi)|^2d\bxi\bigg)^\frac{1}{2} \mathcal{F}_{n+1}(\mu,\nu)\\
    &\le \bigg(\sum_{(\beta_1,\dots,\beta_n)\in J_m} \int_{\erre^n}|\partial^{\beta_1}_{x_1}\dots\partial^{\beta_n}_{x_n}(\phi\chi_R)(\bx)|^2d\bx\bigg)^\frac{1}{2} \mathcal{F}_{n+1}(\mu,\nu)\\
    &\le C R^{\frac{n}{2}}\mathcal{F}_{n+1}(\mu,\nu)
\end{align*}
where $C=C(m,n)$ is a suitable finite constant, and $J_m$ are the multi-index such that $\sum_{j=1}^n\beta_j=m$ where $m = \frac{n+1}{2}$ if $n$ is odd and, and $m= \frac{n+2}{2}$ if $n$ is even. 
Then, by taking the supremum in $\phi$, we conclude that
\begin{equation}
\label{WFourboundR}
   \|f-g\|_m^\star \le CR^{\frac{n}{2}}\mathcal{F}_{n+1}(\mu,\nu)+\frac{M}{R}.
\end{equation}
If we take the derivative with respect to $R$ we have that the minimum of the right hand side is attained when
\[
    R=\bigg(\frac{2M}{nC}\bigg)^{\frac{2}{n+2}}\frac{1}{\mathcal{F}_{n+1}(\mu,\nu)^{\frac{2}{n+2}}}.
\]
By plugging this value in \eqref{WFourboundR}, we find
\begin{equation}
\label{WFourboundR1}
    \|f-g\|_m^\star\le C^{\frac{2}{n+2}}\bigg(\frac{2M}{n}\bigg)^{\frac{n}{n+2}}R^{\frac{n}{2}}\mathcal{F}_{n+1}(\mu,\nu)^{\frac{2}{n+2}}+M^{\frac{n}{n+2}}\bigg(\frac{nC}{2}\bigg)^{\frac{2}{n+2}}\mathcal{F}_{n+1}(\mu,\nu)^{\frac{2}{n+2}},
\end{equation}
which yields
\begin{equation}
   \|f-g\|_m^\star\le C_{M,n}\mathcal{F}_{n+1}(\mu,\nu)^{\frac{2}{n+2}}.
\end{equation}
Substituting the last display into \eqref{eq:CaTo} closes the proof.
\end{proof}

% \textcolor{purple}{[GB: looks nice and quite in the spirit of Nash's inequality]}

\subsubsection{Comparing the Wasserstein Distances and the Energy Distances}

Lastly, we compare the Wasserstein Distance with the other two divergences.
Since the Energy Distance is equal to a suitable Fourier-based Metric up to a constant, it is sufficient to consider the Energy Distance.
Albeit we have shown that these two divergences are, under the right assumptions, equivalent--thus inducing the same notion of convergence--they exhibit distinct behaviours.
Notice that the properties of the Wasserstein Distance are well-known and comprehensively studied\cite{ambrosio2005gradient,santambrogio2015optimal,villani2009optimal}, we briefly recall them for the sake of completeness.

{\textbf{Analytical Properties.}}
As per the Energy Distance, the Wasserstein Distance is continuous and differentiable for every $p\ge 1$.
However, the principal advantage of the Wasserstein Distance is that it is well-defined whenever the two measures possess a sufficient number of finite moments. 
In contrast, the Energy Distance requires the measures to have equal moments up to a certain order.
% 
% Moreover it is continuous and differentiable, as shown in the seminal references Refs. \cite{ambrosio2005gradient,santambrogio2015optimal,villani2009optimal}.
% 

% 

% 
{\textbf{Topological Properties.}}
Unilike the Energy Distance, the Wasserstein Distance is $1$-scale sensitive for every $p$ and sub-additive for convolutions for $p=2$, hence Zolotarev ideal whenever $p=2$.\cite{CaTo}
In particular, the Wasserstein Distance is not scale invariant and, as per the Energy Distance, it is sub-additive by convolution when $p\ge 4$.

{\textbf{Computational Properties.}}
On the computational side, the Wasserstein Distance presents two major drawbacks:
\begin{enumerate*}[label=(\roman*)]
    \item it is computationally expensive, as it involves solving a linear programming problem\cite{NIPSABGV,LingOkada2007} or computing an approximation through convex-optimization methods\cite{cuturi2013sinkhorn,AJJR}; and
    \item it yields biased gradients, as observed in Ref. \cite{belle}, which is an undesirable trait when dealing with machine learning.
\end{enumerate*}
As demonstrated, these issues do not arise with the Energy Distance.

% Lastly, we compare the Wasserstein Distance with the Energy Distance. 
% % 
% Albeit being equivalent and thus inducing the same notion of convergence, the two divergences have different traits.
% % 
% Indeed, the main advantage of using the Wasserstein distance is that it is well-defined whenever the two measures have a suitable number of finite moments, whereas the Energy Distance needs the measures to have equal moment up to a certain order.
% % 
% This aside, the Wasserstein Distance has two major flaws from an applied viewpoint: \begin{enumerate*}[label=(\roman*)]
%     \item it is expansive to compute, as it requires to solve a Linear Programming Problem; and
%     \item it has biased gradient, as noticed in \citeref{}.
% \end{enumerate*}
% % 
% As we have shown, these flaws are not shared by the Energy Distance.
% % 

% % 
% Notice that our considerations trivially extend to the Fourier-based Metrics, as they are up to a multiplicative constant equal to a suitable Energy Distance.

\section{Scale Invariant Divergences}

None of the divergences examined in this study satisfy the scale invariance property as defined in Definition \ref{def:scaleinvariance}.
However, as argued in Ref. \cite{ABGT}, one can enforce scale invariance on a divergence $\mathcal{D}$ by applying a suitable whitening transformation to the input probability measures as a pre-processing step.
This procedure yields a new whitened divergence that is, by construction, scale invariant. Nevertheless, this transformation may compromise certain desirable properties originally possessed by $\mathcal{D}$.
In this section, we identify what properties of $\mathcal{D}$ are preserved by its whitened counterpart and which ones are not compatible with scale invariant divergences.

\subsection{The Whitening Process}
\label{sec:whitening_processes}
In this section, we recall the definition of whitening process and its properties.
In what follows, we denote with $\mathcal{P}_{\Id}(\erre^n)$ the subset of $\mathcal{P}(\erre^n)$ containing the probability measures whose covariance is the $n\times n$ identity matrix $\Id$.
% 
% \textcolor{purple}{[GB: find a special symbol for the identity.]}
% 

\begin{definition}
    A \textit{whitening process} is a function $\mathcal{S}$ that maps any $n$-dimensional random vector $\bX$ to a $n$-dimensional random vector whose covariance is the identity matrix, thus $\mathcal{S}:\mathcal{P}(\erre^n)\to\mathcal{P}_{\Id}(\erre^n)$. 
    A whitening process $\mathcal{S}$ is linear if, for any given $\bX\in\mathcal{P}(\erre^n)$, we have
    \begin{equation}
        \label{whi1}
        {\bX}^{*} = \mathcal{S}(\bX) = W_{\bX}{\bX} = W_{\mu}{\bX},
    \end{equation}
    where $W_\bX=W_\mu$ is a $n\times n$ matrix that depends on $\bX\sim\mu$.
    The matrix $W_\mu$ is also known as whitening matrix (associated with $\mathcal{S}$).\cite{KLS}
\end{definition}

If the covariance matrix of $\bX$, namely $\Sigma_\mu$, is invertible the whitening matrix in \eqref{whi1} must satisfy the identity $W_\mu\Sigma_\mu W_\mu^T = \Id$ and thus $W_\mu (\Sigma_\mu W_\mu^TW_\mu) = W_\mu$, which boils down to
\begin{equation}
    \label{whi2}
    W_\mu^TW_\mu =\Sigma_\mu^{-1}.
\end{equation}

Notice that condition \eqref{whi2} does not determine a unique matrix $W_\mu$.
Indeed, if $W_\mu$ satisfies equation \eqref{whi2}, any matrix of the form $\widetilde W_\mu = Z W_\mu$ with $Z$ such that $Z^TZ=\Id$ also satisfies \eqref{whi2}.
Consequently, there is no unique linear whitening process.\cite{LZ}

An important property that a whitening processes $\mathcal{S}$ should possess is \emph{Scale Stability}, which ensures that the output of $\mathcal{S}$ does not change if the unit of measurement of the input random vector changes.

\begin{definition}[Scale Stability]
    A whitening process $\mathcal{S}$ is \emph{Scale Stable} if
    \begin{equation}
        \label{eq:scale_stab}
        \mathcal{S}({\bX})=\mathcal{S}(Q{\bX}),
    \end{equation}
    for every random vector $\bX$ and for any diagonal matrix $Q=\diag(q_1,q_2,\dots,q_n)$
    % \[
    %     Q=\diag(q_1,q_2,\dots,q_n):=\begin{pmatrix}
    %         q_1 & 0 & 0 & \dots & 0\\
    %         0 & q_2 & 0 & \dots & 0\\
    %         0 & 0 & q_3 & \dots & 0\\
    %         \vdots & \vdots & \vdots & \ddots & \vdots\\
    %         0 & 0 & 0 & \dots & q_n\\
    %     \end{pmatrix}
    % \]
    with $q_i>0$ for every $i=1,\dots n$.
\end{definition}

A \emph{Scale Stable} whitening processes exist: both the \textit{Cholesky} whitening and the \textit{Zero-Components Analysis} of the correlation matrix (ZCA-cor) whitening possess this property.\cite{AGT}\\

The first whitening process we consider is the\textit{Cholesky Whitening}, which is based on the Cholesky factorization of a positive-definite matrix. 
In this case, the whitening matrix associated with $\bX\sim\mu$ is defined as 
\begin{equation}
\label{Chol}
W_\mu^{Chol} =L^T, 
\end{equation}
where $L$ is the unique lower triangular matrix with positive diagonal values that satisfies \eqref{whi2}. 
Owing to the triangular structure of $W_\mu^{Chol}$ and to the uniqueness of the decomposition, the Cholesky whitening process is Scale Stable.

\begin{theorem}
\label{thm:CholSS}
    The Cholesky whitening process is {Scale Stable}.
\end{theorem}

\begin{proof}
    We refer the reader to Ref. \cite{AGT} for the proof.
\end{proof}

The second whitening process we consider is the \textit{Zero-Components Analysis} of the correlation matrix (ZCA-cor), which is derived from the correlation matrix associated with the random vector $\bX$.
Given a random vector $\bX$, let $P$ denote the covariance matrix associated with $\bX$.
Then, the ZCA-cor whitening matrix associated with $\bX$ is defined as
\begin{equation}
\label{ZCA-cor}
W_\mu^{ZCA} = P^{-1/2}V^{-1/2},
\end{equation}
where $P^{-1/2}$ is a square root of the inverse matrix of $P$ and 
\[
    V=\diag(Var(X_1),Var(X_2),\dots,Var(X_n))=\begin{pmatrix}
        Var(X_1) &0  &\dots &0\\
        0 &Var(X_2) &\dots &0\\
        \vdots &\vdots &\ddots &\vdots\\
        0 &0 &\dots &Var(X_n)
    \end{pmatrix}
\]
is a diagonal matrix whose diagonal contains the variances of the entries of $\bX$.
% 

% 
% \[
% V=\begin{pmatrix}
%     Var(X_1) & 0  & \dots & 0\\
%     0 & Var(X_2) & \dots & 0\\
%     \vdots & \vdots & \ddots &\vdots \\
%     0 & 0 &\dots & Var(X_n)
% \end{pmatrix}.
% \]
% 

% 
Again, notice that $P^{-\frac{1}{2}}$ in \eqref{ZCA-cor} is not defined uniquely, thus there are multiple ZCA-cor matrices associated with the same $\bX\in\mathcal{P}(\erre^n)$.
Since the correlation matrix of $\bX$ and $Q\bX$ are the same for every diagonal matrix $Q$, the ZCA-cor whitening process is scale stable once we decide how to compute the square root of $P$.

\begin{theorem}
\label{thm:ZCAcorrSS}
    The ZCA-cor whitening process is {Scale Stable}, regardless of how the square root of $P^{-1}$ is selected.
\end{theorem}

\begin{proof}
    Let us fix $IS$ a function that assigns to each matrix one of the possible square roots, so that $IS(M)IS(M)=M$ for any invertible matrix $M$, and define
    \[
        \mathcal{S}(\bX)=IS(P^{-1})V^{-\frac{1}{2}}\bX.
    \]
    Given $Q=\diag(q_1,q_2,\dots,q_n)$ a diagonal matrix whose diagonal values are positive, we have that the correlation matrix of $Q\bX$ and the correlation matrix of $\bX$ coincide.
    In particular, we have that
    \[
        \mathcal{S}(Q\bX)=IS(P^{-1})V^{-\frac{1}{2}}Q^{-1}(Q\bX)=IS(P^{-1})V^{-\frac{1}{2}}\bX=\mathcal{S}(\bX),
    \]
    hence $\mathcal{S}$ is scale stable.
\end{proof}

Theorem \ref{thm:ZCAcorrSS} allows us some degree of freedom in defining the ZCA-cor whitening process.
In the remainder of the paper, we consider any ZCA-cor whitening process for which it holds $W^{ZCA}_\mu=\Id$ whenever $\mu\in\PP_{\Id}(\erre^n)$.
Notice that the Cholesky whitening process naturally possesses this property, that is $W^{Chol}_\mu=\Id$ if $\mu\in\PP_{\Id}(\erre^n)$.
Now that we have established that a scale stable whitening process always exists, we are ready to introduce the notion of whitened divergence.

\begin{definition}
    Let $\mathcal{D}:\PP(\erre^n)\times\PP(\erre^n)\to[0,\infty)$ be a divergence and let $\mathcal{S}$ be a whitening process. 
    We define the whitened divergence $\DDSS$ associated with $\mathcal{S}$ as
    \[
    \DDSS(\bX,\bY)=\mathcal{D}(\mathcal{S}(\bX),\mathcal{S}(\bY)).
    \]
\end{definition}

% As a first key result, we show that if $\mathcal{S}$ is scale stable, then $\DDSS$ is scale invariant.

\begin{theorem}
    \label{thm:scale_invariance_via_whitening}
    If $\mathcal{S}$ is scale stable, then $\DDSS$ is scale invariant.
\end{theorem}

\begin{proof}
    Let $Q$ and $Q'$ be two diagonal matrix with positive diagonal values.
    Then, given two random vectors $\bX$ and $\bY$, we have
    \[
        \DDSS(Q\bX,Q'\bY)=\mathcal{D}(\mathcal{S}(Q\bX),\mathcal{S}(Q'\bY))=\mathcal{D}(\mathcal{S}(\bX),\mathcal{S}(\bY))=\DDSS(\bX,\bY),
    \]
    which concludes the proof.
\end{proof}

For the sake of simplicity, in what follows we consider the ZCA-cor whitening process, unless we specify otherwise.
Notice however, that all our conclusions can be equally inferred for the Cholesky whitening process.

\begin{remark}
    % Owing to the notion of whitened divergence, we are able to extend the connection between the Energy Distance and the Gini Index pointed out in Remark \ref{rmrk:gini_energy_onedim}.
    % 
    % When extending the one dimensional Gini index in \eqref{eq:def_gini} to the higher dimensional case, multiple approaches are available; however, b
    % 
    By selecting an appropriate Energy Distance, it is possible to extend the connection between the Energy Distance and the Gini Index highlighted in Remark \ref{rmrk:gini_energy_onedim}.
    In this remark we limit our study to two cases, which correspond to the classic Energy Distance $\mathcal{E}_1$ (defined as in \eqref{en-dist}) and the $l_1$ Energy Distance $\mathcal{E}^{l_1}_1$ (studied in Ref. \cite{AGTB:IJCNN}).
    \textbf{Connection between the Gini Index and $\mathcal{E}_1$.}
    The first higher dimensional extension of the Gini Index we consider is $\mathcal{G}_T$, introduced in Ref. \cite{giudici2024measuring}, and defined as  
    \begin{equation}
        \mathcal{G}_T(\bX)=\frac{1}{2\sqrt{\bm_{\bX}^T\Sigma^{-1}_\mu\bm_{\bX}}}\int_{\erre^n\times\erre^n}\sqrt{(\bx-\bx')^T\Sigma^{-1}(\bx-\bx')}d\mu(\bx)d\mu(\bx').
    \end{equation}
    Within this framework, given $\bX\sim\mu$ and $\bY\sim\nu$ two random vectors, we have
    \begin{align*}
        \mathcal{E}_{1,ZCA}(\bX,\bY)&=2\GMD(W^{ZCA}_\mu\bX^*,W^{ZCA}_\nu\bY^*)\\
        &\quad-2\sqrt{\bm_{\bX}^T\Sigma^{-1}_\mu\bm_{\bX}}\mathcal{G}_T(\bX)-2\sqrt{\bm_{\bY}^T\Sigma^{-1}_\nu\bm_{\bY}}\mathcal{G}_T(\bY),
    \end{align*}
    where $\GMD$ is the Gini Mean Difference for higher dimensional random variables and $W^{ZCA}_\gamma$ is the ZCA whitening matrix associated with $\gamma$.

    \textbf{Connection between the Gini Index and $\mathcal{E}^{l_1}_1$.} 
    The second higher dimensional extension of the Gini Index we consider is $\mathcal{G}_{l_1}$, introduced in Ref. \cite{AGT}, and defined as
    \begin{equation}
        \mathcal{G}_{l_1}(\bX)=\frac{1}{2|W^{ZCA}_\mu\bm_\mu|_{l_1}}\int_{\erre^n\times\erre^n}|W^{ZCA}_\mu(\bx-\bx')|_{l_1}d\mu(\bx)d\mu(\bx'),
    \end{equation}
    where $|\bx|_{l_1}$ is the $l_1$ norm of the vector $\bx$, that is $|\bx|_{l_1}=\sum_{i=1}^n|x_i|$, and $W^{ZCA}_\mu$ is the ZCA whitening matrix associated with the probability distribution $\mu$.
    If we now consider the Energy Distance induced by the $l_1$, that is
    \begin{align*}
        \mathcal{E}^{l_1}_1(\bX,\bY)=&2\int_{\R^{2n}}|\bx-\by|_{l_1} d\mu(\bx) d\nu(\by)\\
        &\;\;- \int_{\R^{2n}}|\bx-\by|_{l_1} d\mu(\bx) d\mu(\by) - \int_{\R^{2n}}|\bx-\by|_{l_1} d\nu(\bx) d\nu(\by),
    \end{align*}
    we then have that
    \begin{equation}
        \mathcal{E}^{l_1}_{1,ZCA}(\bX,\bY)=2\sum_{i=1}^n\GMD(X^*_i,Y^*_i)-2|W_\mu^{ZCA}\bm_\bX|_{l_1}\mathcal{G}_{l_1}(\mu)-2|W_\nu^{ZCA}\bm_\bY|_{l_1}\mathcal{G}_{l_1}(\nu),
    \end{equation}
    where $\GMD(X^*_i,Y^*_i)$ is the one dimensional Gini Mean Difference between the whitened components of the random vectors $\bX$ and $\bY$.
\end{remark}

\subsection{The Relation between Divergences and Whitened Divergence}

Theorem \ref{thm:scale_invariance_via_whitening} gives us a direct way to retrieve a scale invariant divergence starting from any divergence function $\mathcal{D}$.
In this section, we study whether the properties of the original divergence $\mathcal{D}$ are preserved by its whitened counterpart.

\subsubsection{The Analytic Properties}

First, we study the metric properties of the divergence $\DDSS$.
It is easy to see that any whitened divergence cannot be a metric between probability distributions.
% 
% Indeed, given any divergence $\mathcal{D}$ and a scale stable whitening $\mathcal{S}$, we have that $\DDSS$ cannot be a metric.
% 
Indeed, by definition, we have $\DDSS(\bX,c\bX)=0$ for every $c>0$ hence,
% 
% , thus violating the first axiom of metrics.
% 
in general, a scale invariant metric cannot satisfy the first axiom that defines a metric.
Albeit $\DDSS$ is not a metric, we show that it is a metric when the divergence $\mathcal{D}$ is a metric and we restrict $\DDSS$ to a suitable quotient space of the probability measures.
Indeed, any whitening process $\mathcal{S}$ induces a equivalence relation on the space of probability distributions $\PP(\erre^n)$.
In what follows, we say that $\bX\simS\bY$ if we have that $\mathcal{S}(\bX)=\mathcal{S}(\bY)$.
% \[
% \mathcal{S}(\bX)=\mathcal{S}(\bY).
% \]

\begin{theorem}
    \label{thm:equivsims}
    Given a whitening process $\mathcal{S}$, the relation $\simS$ induces a equivalence relation on $\PP(\erre^n)$ if 
    \begin{equation}
        \label{eq:stable_whitening}
        \mathcal{S}(\bX)=\bX
    \end{equation}
    for every $\bX\in\PP_{\Id}(\erre^n)$. Moreover, we have that $\PP(\erre^n)/\simS=\PP_{\Id}(\erre^n)$.
\end{theorem}

\begin{proof}
    Owing to the fact that $\mathcal{S}(\bX)=\bX$ for every $\bX\in\PP_{\Id}$, the first statement of the Theorem follows trivially.
    Let us consider a class of equivalence of $\PP(\erre^n)/\simS$, namely $[\bZ]$.
    By definition of $\mathcal{S}$, we have that $\mathcal{S}(\bZ)\in\PP_{\Id}(\erre^n)$.
    Moreover, we have that $\mathcal{S}(\bZ)\simS\bZ$, hence $\mathcal{S}(\bZ)$ is a representative of the equivalence class to which $\bZ$ belongs.
    To conclude, we show that two different elements of $\PP_{\Id}(\erre^n)$, namely $\bX$ and $\bX'$, induce two different classes of equivalence as long as $\bX\neq\bX'$.
    By the hypothesis on $\mathcal{S}$, we have that $\mathcal{S}(\bX)=\bX$ whenever $\bX\in\PP_{\Id}(\erre^n)$.
    Thus we have that $\mathcal{S}(\bX)=\bX\neq\bX'=\mathcal{S}(\bX')$, which concludes the proof.
\end{proof}

% 

% \begin{remark}
%     Notice that both the ZCA-cor and the Cholesky whitening processes satisfy condition \eqref{eq:stable_whitening}.
% \end{remark}

% % 
% We are then ready to present our main result of the section, which states that if $\mathcal{D}$ is a metric between probability measures, then $\DDSS$ is a metric between classes of equivalences as long as $\mathcal{S}$ satisfies condition \eqref{eq:stable_whitening}.
% % 

\begin{theorem}
    Let $\mathcal{D}:\PP(\erre^n)\times\PP(\erre^n)\to[0,\infty)$ be a metric between probability distribution and let $\mathcal{S}$ be a scale stable whitening process such that $\mathcal{S}(\bX)=\bX$ for every $\bX\in\PP_{\Id}(\erre^n)$.
    We then have that $\DDSS$ is a metric over $\PP(\erre^n)/\simS$.
    % , defined as 
   %  \[
   %  \DDSS(\bX,\bY)=\mathcal{D}(\mathcal{S}(\bX),\mathcal{S}(\bY)),
   %  \]
   % is a metric over $\PP(\erre^n)/\simS$.
\end{theorem}

\begin{proof}
    By Theorem \ref{thm:equivsims}, we have $\PP(\erre^n)/\sim_\mathcal{S}=\PP_{\Id}(\erre^n)$.
    Given two classes of equivalence $[\bX],[\bY]\in\PP(\erre^n)/\sim_\mathcal{S}$, we then have
    \[
        \DDSS([\bX],[\bY])=\mathcal{D}(\mathcal{S}(\bX),\mathcal{S}(\bY))=\mathcal{D}(\bX,\bY).
    \]
    In particular, we infer that $\DDSS$ on $\PP(\erre^n)/\sim_\mathcal{S}$ is equal to $\mathcal{D}$ restricted to $\PP_{\Id}(\erre^n)$, hence the thesis.
\end{proof}

\subsubsection{The Topological Properties} 

In this section, we consider the topological properties of the whitened divergence $\DDSS$ and show that the equivalence between two divergences is preserved by the whitening pre-processing.

\begin{theorem}
    If $\mathcal{D}$ and $\mathcal{D}'$ are equivalent, then also $\DDSS$ and $\DDSS'$ are equivalent.
    In particular, if $c_l\mathcal{D}^s(\bX,\bY)\le \mathcal{D}'(\bX,\bY)\le c_u\mathcal{D}^r(\bX,\bY)$ then, for every scale stable whitening process we have
    \begin{equation}
        c_l\DDSS^s(\bX,\bY)\le \DDSS'(\bX,\bY)\le c_u\DDSS^r(\bX,\bY).
    \end{equation}
\end{theorem}

\begin{proof}
    By definition, we have that
    \begin{align*}
        \DDSS'(\bX,\bY)&=\mathcal{D}'(\mathcal{S}(\bX),\mathcal{S}(\bY))\le c_u\mathcal{D}^r(\mathcal{S}(\bX),\mathcal{S}(\bY))=c_u\DDSS^r(\bX,\bY).
    \end{align*}
    Likewise, it is possible to show $c_l\DDSS^s(\bX,\bY)\le\DDSS'(\bX,\bY)$.
\end{proof}

Lastly, notice that each $\DDSS$ is scale invariant thus cannot be scale sensitive, however the sub-additivity by convolution is preserved if $\mathcal{D}$ is Zolotarev ideal.

\begin{theorem}
    \label{thm:ZI}
    Let $\mathcal{S}$ be the ZCA-cor or the Cholesky whitening process.
    If $\mathcal{D}$ is Zoloratev ideal, then $\DDSS$ restricted over $\PP_{\Id}(\erre^n)$ satisfies the following condition
    \begin{equation}
        \label{eq:simple_subadd}
        \DDSS(\bX+\bZ,\bY+\bW)\le \DDSS(\bX,\bY) + \DDSS(\bZ,\bW)
    \end{equation}
    for every quadruplet of random vectors $\bX,\bY,\bZ$, and $\bW$ such that $\bX$ and $\bZ$ are independent, $\bY$ and $\bW$ are independent.
\end{theorem}

\begin{proof}
    Let $\bX,\bY,\bZ$, and $\bW$ be four random vectors such that $\bX$ and $\bZ$ are independent, $\bY$ and $\bW$ are independent.
    Since we are working on the quotient space $\PP(\erre^n)/\sim_\mathcal{S}=\PP_{\Id}(\erre^n)$, the covariance matrices of the four random vectors are the identity matrix.
    Owing to the independence assumption, we have that $\bX+\bZ\sim\eta$ has a covariance matrix equal to $2\Id$, thus we have that $W^{ZCA}_\eta=W^{Chol}_\eta=\frac{1}{\sqrt{2}}\Id$.
    Likewise, we have that both the Cholesky and the ZCA-cor whitening matrix associated with $\bY+\bW$ is equal to $\frac{1}{\sqrt{2}}\Id$.
   Since $\mathcal{D}$ is Zolotarev ideal, we have that
    \begin{align*}
        \DDSS(\bX+\bZ,\bY+\bW) &= \mathcal{D}\Big(\mathcal{S}(\bX+\bZ),\mathcal{S}(\bY+\bW)\Big) = \mathcal{D}\Big(\frac{1}{\sqrt{2}}(\bX+\bZ),\frac{1}{\sqrt{2}}(\bY+\bW)\Big)\\
        &\le \mathcal{D}\Big(\frac{1}{\sqrt{2}}\bX,\frac{1}{\sqrt{2}}\bY\Big) + \mathcal{D}\Big(\frac{1}{\sqrt{2}}\bZ,\frac{1}{\sqrt{2}}\bW\Big)\\
        &=\frac{1}{\sqrt{2}}\Big(\mathcal{D}(\bX,\bY) + \mathcal{D}(\bZ,\bW)\Big)\le \mathcal{D}(\bX,\bY) + \mathcal{D}(\bZ,\bW)\\
        &=\DDSS(\bX,\bY)+\DDSS(\bZ,\bW),
    \end{align*}
    which concludes the proof.
\end{proof}

\begin{remark}
    Notice that the proof of Theorem \ref{thm:ZI} works even under the assumption that $\mathcal{D}$ is sub-additive by convolution and it is $p$-scale sensitive for any $p\ge 0$.
\end{remark}

\subsubsection{The computational Properties}

To conclude, we consider the computational properties of the whitened divergences.
First, we show that if the divergence $\mathcal{D}$ has an unbiased gradient, the also $\DDSS$ has unbiased gradient \eqref{def:unbiasedgrad}.
% 
% \textcolor{purple}{[GB: the proof does not use the fact that we are actually measuring the gradient, right? It could have been any other operation, provided it is well-posed also on the image measures $\nu^\star_\theta$. Do I get it correctly?]}

\begin{theorem}
\label{thm:preserveunbiasedgrad}
    Let $\mathcal{D}:\PP(\erre^n)\times\PP(\erre^n)\to[0,\infty)$ be a divergence with unbiased gradient.
    Let $\mathcal{S}$ be a scale stable whitening transformation that preserves differentiability; that is, for any differentiable parametric family of probability distributions $\{\nu_\theta\}_{\theta \in I}$, the image family through $\mathcal{S}$, namely $\{\nu_\theta^*\}_{\theta \in I}$, is also differentiable with respect to $\theta$.
    Then, the whitened divergence $\DDSS$ has unbiased gradient.
\end{theorem}

\begin{proof}
Let $\mathcal{S}$ be a whitening process as in the hypothesis and let $\bX \sim \mu$ be a random vector and its associated probability distribution.
We denote by $\mu^*$ the distribution of the whitened random vector $\mathcal{S}(\bX)$. 
Similarly, consider a parametric family of probability distributions $\bY_\theta \sim \nu_\theta$, and let $\nu^*_\theta$ denote the distribution of the whitened random variable $\mathcal{S}(\bY_\theta)$.
Assume that the divergence $\mathcal{D}$ has unbiased gradient.
Then, owing to the fact that $\{\nu_\theta^*\}_{\theta\in I}$ is differentiable, the following equality holds
\begin{align*}
    \nabla_{\theta} \DDSS(\mu, \nu_\theta) &= \nabla_{\theta} \mathcal{D}(\mu^*, \nu_\theta^*)\\
    &= \mathbb{E}_{\bX^{(1)}, \dots, \bX^{(N)}}\left[ \nabla_{\theta} \mathcal{D}({\mu}_N^*, \nu_\theta^*) \right]\\
    &= \mathbb{E}_{\bX^{(1)}, \dots, \bX^{(N)}}\left[ \nabla_{\theta} \DDSS({\mu}_N, \nu_\theta) \right],
\end{align*}
which concludes the proof.
\end{proof}

\begin{remark}
    If the family $\nu_\theta$ is composed of absolutely continuous distributions, we then have that the ZCA-cor whitening process satisfies the conditions in Theorem \ref{thm:preserveunbiasedgrad} if there exists $\epsilon>0$ such that $|det(\Sigma_\theta^{-1})|\ge \epsilon$, where $\Sigma_\theta^{-1}$ is the inverse covariance matrix of $\nu_\theta$. 
\end{remark}

Finally, we comment on the computational complexity of $\DDSS$.
Both the ZCA-cor and Cholesky whitening processes involve matrix inversion followed by either a square root or Cholesky factorization—operations that require $\mathcal{O}(n^3)$ computations, where $n$ is the matrix dimension.
Importantly, this cost is independent of the number of points in the support of the probability distributions being compared.
% 
% 
% Lastly, we comment on the computational complexity of $\DDSS$.
% % 
% Indeed, we have that computing the ZCA-cor whitening process and Cholesky whitening process requires to compute the inverse of a matrix and then either taking the square root or computing its Cholesky factorization.
% % 
% All these operations are well studied and their completion requires $Cn^3$ operations, where $n$ is the dimension of the matrix and $C$ is a suitable constant.
% % 
% In particular, computing the ZCA-cor or the Cholesky whitening processes require a number of operations that does not depend on the number of points in the support of the probability distributions under study.
% 

\begin{theorem}
    Let $\mathcal{S}$ be the ZCA-cor whitening process or the Cholesky whitening process.
    Then, the divergence $\mathcal{D}$ has the same complexity of $\DDSS$. 
\end{theorem}

\begin{proof}
    Let $\mathcal{D}$ be a divergence whose complexity is $O(N^p)$ for a suitable $p\ge 0$.
    To compute $\DDSS$, we need to compute either the ZCA-cor whitening matrix or the Cholesky whitening matrix, which take $O(1)$ operations.
    Owing to the properties of the Landau operators $O$, we have that $O(N^p)+O(1)=O(N^p)$ allowing us to conclude the proof.
\end{proof}

\section{Application}

\label{sec:examples}

\begin{table*}[t]
    \centering
    \begin{adjustbox}{width=0.9\textwidth}
    \footnotesize
    \begin{tabular}{llll}
        \toprule
        Sector & Detail & Frequency & Percentage \\ 
\midrule
        Consumer & Consumer Staples \& Consumer Discretionary & 351 & 33.05 \\ 
        Financials & Banking, Insurance \& Financial Services, and Real Estate & 14 & 1.32 \\ 
        Health.Util & Healthcare \& Essential Utility Services & 55 & 5.18 \\  
        Manufacturing & Materials \& Industrial Activities & 579 & 54.52 \\
        Tech.Com & Information Technology \& Communication Industries & 63 & 5.93 \\ 
        \midrule
        Total & & 1062 & 100.00 \\
        \bottomrule
    \end{tabular}
    \end{adjustbox}
    \caption{Distribution of Companies by Sector Classification.}
    \label{Sector-Classification}
\end{table*}

In this Section we exemplify the  energy divergence by means of its application to an important economic problem: the impact of Environmental, Social, and Governance (ESG) factors on the financial growth of companies.

We consider annual data from a set of Small and Medium Enterprises across various sectors, in the period from $2020$ to $2022$.

The data, provided by modefinance, a fintech rating agency, includes the ESG scores of the considered companies) along with a set of key financial metrics.  
The data also contain the classification of each company in an economic sector, following the Global Industry Classification Standard (GICS), as shown in Table \ref{Sector-Classification}.

The dataset consists of a total of $1,062$ observations, corresponding to  $1,062$ different Small and Medium Enterprises (SMEs).
 Table \ref{Sector-Classification} displays the distribution of the SMEs across the different sectors. 
% 
%Out of a total of $1,062$ companies, $351$ ($33.05\%$) are from the Consumer sector, $14$ ($1.32\%$) from the Financial sector, $55$ ($5.18\%$) from the Health \& Essential Utilities sector, $579$ ($54.52\%$) from the Manufacturing sector, and $63$ ($5.93\%$) from the Technology \& Communication Industries sector.
% 
 Table \ref{Summary Statistics-ESG} reports a summary analysis of the ESG metrics for the considered companies.

\begin{table}[!ht]
    \centering
    \footnotesize
    \begin{tabular}{lcccc}
        \toprule
        & ESG & E.Sc & S.Sc & G.Sc \\
        \midrule
        Mean   & 0.65 & 0.76 & 0.51 & 0.62 \\
        Median & 0.66 & 0.79 & 0.50 & 0.64 \\
        Sdev   & 0.11 & 0.17 & 0.23 & 0.15 \\
        Min    & 0.28 & 0.07 & 0.07 & 0.21 \\
        Max    & 0.93 & 0.93 & 0.93 & 0.93 \\
        Range  & 0.65 & 0.86 & 0.86 & 0.71 \\
        \bottomrule
    \end{tabular}
    \caption{Summary Statistics for ESG Metrics.}
    \label{Summary Statistics-ESG}
\end{table}

From Table \ref{Summary Statistics-ESG}, note that the Overall ESG Scores (ESG) exhibits some variability.
Environmental Scores (E.Sc) are  higher, on average ($0.76$).  Social Scores (S.Sc) are considerably lower, and  with higher variability. Finally, 
Governance Scores (G.Sc) have a lower  variability, with an average score of $0.62$.
The financial performance of the companies is measured by three indicators, described in 
Table \ref{Variable-Definitions}, with summary statistics in Table \ref{Summary Statistics-FI} for the year $2022$.

\begin{table}[!ht]
    \centering
    \footnotesize
    \begin{tabular}{cllr}
        \toprule
        & Indicator & Code & Description \\ 
        \midrule
        1 & Total Asset & TASS & Size of assets. \\ 
    
        2 & Turnover & TOVR &  Size of Sales. \\ 
        
        3 & Shareholders' Funds & SFND & Size of Equity. \\ 
    
        \bottomrule
    \end{tabular}
    \caption{Description of Key Financial Metrics.}
    \label{Variable-Definitions}
\end{table}

\begin{table}[!ht]
    \centering
    \footnotesize
    \begin{tabular}{lrrr}
        \toprule
        & TASS & SFND & TOVR \\
        \midrule
        Mean   & 173476.54 & 62341.83  & 170644.62 \\
        Median & 43824.91  & 15054.90  & 43528.31  \\
        Sdev   & 702804.70 & 250588.41 & 602012.08 \\
        Min    & 1151.23   & -49091.00 & 1288.97   \\
        Max    & 14392422.00 & 5336752.00 & 10587145.00 \\
        Range  & 14391270.77 & 5385843.00 & 10585856.03 \\
        \bottomrule
    \end{tabular}
    \caption{Summary Statistics for Financial Indicators - 2022 Annual Data.}
    \label{Summary Statistics-FI}
\end{table}

Table \ref{Summary Statistics-FI} shows that  Total Assets (TASS) exhibit a broad range that goes from $1,151.23$ EUR to $14,392,422.00$ EUR, with a high standard deviation.
Shareholders' Funds (SFND) and Turnover (TOVR) also have a high standard deviation. These results reflect significant differences in the financial growth of the considered companies.
Standard calculations show that the   sustainability indicators and financial indicators are weakly correlated, with a maximum correlation of $0.15$. 
This indicates that it will be quite challenging to build a statistical model that can predict financial variables based on sustainability variables. It will be extremely important to choose a divergence function that can lead to the truly best model. 
Without loss of generality, we consider four alternative statistical models to predict the three dimensional financial performance response based on the three dimensional vector of ESG score predictors. Specifically:
\begin{enumerate*}[label=(\roman*)]
    \item a multivariate regression model, in which the financial performance variables are explained by theESG scores, independently of the sector (LIN);
    \item a similar model, but dependent on sectors (LINS);  
    \item a neural network model with the same variables as LIN (NNET);
     \item a neural network model with the same variables as LINS (NNETS).
    
\end{enumerate*} 
To compare the four models, we apply the whitening transformation and, then, we randomly split the whitened data into a $80\%$ training sample and a $20\%$ test sample. 

We then compare the four models in terms of the energy divergence. To corroborate our findings, we will consider three different specifications of the exponent of the energy distance: $\alpha=1$, consistent with a multivariate Gaussian distribution; 
$\alpha=0.5$, consistent with a heavy tailed distribution, such as a Pareto; $\alpha=1.5$, consistent with a highly concentrated distribution.
For comparison, we will also consider the standard  root mean squared error metric (RMSE).
% 

% sed euclidean distance.
% 
% 

% 

% 

\begin{table}[!ht]
    \centering
    \begin{tabular}{lcccc}
        \toprule
       & $\alpha=0.5$ & $\alpha=1$& $\alpha=1.5$   & RMSE\\
        LIN  & 93.76 & 71632.86 & 43261969& 570433.01 \\
        LINS & 69.71 & 50374.21& 25867451& 341261.70 \\
        NNET & 138.11 & 55868.11 & 32646035 &596600.20\\
        NNETS & 141.24 &62410.47& 46044578& 359233.40\\
        \bottomrule
    \end{tabular}
    \caption{Comparison of divergences of the four considered models (LIN, LINS, NNET, NNETS), using the Energy Distance for three different values of $\alpha$, as well as  the Root Mean Squared Error. }
    \label{Performances}
\end{table}

In Table \ref{Performances}, we present the comparison of the divergences obtained with the four models, learned on the training set, and utilized to predict the true observations in the test set. 
Table \ref{Performances} shows that  the Energy Divergence reaches a minimum for the LINS model, consistently for all values of $\alpha$. This is also in line with the result obtained  using the RMSE.  We remark that the energy divergence is, differently from RMSE, invariant under scaling: if the ESG values are transformed in another unit of measurement RMSE varies whereas the Energy Distance does  not.
Additionally, the magnitude of the model divergences is monotone in $\alpha$, suggesting that the considered values are sufficient, from a practical viewpoint, to assess the robustness of model choice.  
From an applied point of view, the results indicate that, given the weak influence of ESG variables on financial performance, a simpler linear model makes better predictions. However, such model should include the dependence on sectors: ESG impacts financial growth differently in different sectors.
\section*{Acknowledgment}
%%%%%%%%%%%%%%%%%%%%%%%%%

This work has been written within the activities of GNCS and GNFM groups of INdAM (Italian National Institute of High Mathematics). G.B.~has been funded by the European Union’s Horizon 2020 research and innovation programme under the Marie Sklodowska-Curie grant agreement No 101034413. P.G.  has been funded by the European Union - NextGenerationEU, in the framework of the GRINS- Growing Resilient, INclusive and Sustainable (GRINS PE00000018).

%%%%%%%%%%%%%%%%%%%%%%%%%

%%%%%%%%%%%%%%%%%%%%%%%%%%%%%%%%%%%%%%%%%%%%

%%%%%%%%%%%%%%%%%%%%%%%%%%%%%%%%%%%%%%%%%%%%
\end{document}